\documentclass[pdflatex,sn-chicago]{sn-jnl}

\usepackage{graphicx}%
\usepackage{multirow}%
\usepackage{amsmath,amssymb,amsfonts}%
\usepackage{amsthm}%
\usepackage{mathrsfs}%
\usepackage[title]{appendix}%
\usepackage{xcolor}%
\usepackage{makecell}%
\usepackage{textcomp}%
\usepackage{manyfoot}%
\usepackage{booktabs}%
\usepackage{algorithm}%
\usepackage{algorithmicx}%
\usepackage{algpseudocode}%
\usepackage{listings}%
\usepackage{tabularx}%
\usepackage{float}%
\definecolor{DarkGreen}{RGB}{0, 100, 0}

\theoremstyle{thmstyleone}%
\newtheorem{theorem}{Theorem}
\newtheorem{proposition}{Proposition}%
\newtheorem{assumption}{Assumption}
\newtheorem{lemma}{Lemma}

\theoremstyle{thmstyletwo}%
\newtheorem{remark}{Remark}%

\theoremstyle{thmstylethree}%
\newtheorem{definition}{Definition}%

\def\beq{\begin{equation}}
\def\eeq{\end{equation}}
\def\beqr{\begin{eqnarray}}
\def\eeqr{\end{eqnarray}}
\def\beqrs{\begin{eqnarray*}}
\def\eeqrs{\end{eqnarray*}}
\def\bet{\begin{theorem}}
\def\eet{\end{theorem}}
\def\bel{\begin{lemma}}
\def\eel{\end{lemma}}
\def\bed{\begin{definition}}
\def\eed{\end{definition}}
\def\bep{\begin{proposition}}
\def\eep{\end{proposition}}
\def\bg{\begin{figure}[tbph]\begin{center}}
\def\eg{\end{center}\end{figure}}

\def\bc{\begin{center}}
\def\ec{\end{center}}

\def\mR{\mathbb{R}}

\def\argmin{\mbox{argmin}}

\def\beq{\begin{equation}}
\def\eeq{\end{equation}}
\def\beqr{\begin{eqnarray}}
\def\eeqr{\end{eqnarray}}
\def\beqrs{\begin{eqnarray*}}
\def\eeqrs{\end{eqnarray*}}

\begin{document}

\title[Article Title]{Bi-SCORE for Weighted Bipartite Networks with Application in Knowledge Source Discovery}


\author[1]{\fnm{Zicheng} \sur{Xie}}\email{zichengxie753@gmail.com}
\equalcont{Zicheng Xie and Rui Pan are co-first authors and contributed equally to this work.}

\author[1]{\fnm{Rui} \sur{Pan}}\email{panrui\_cufe@126.com}
\equalcont{Zicheng Xie and Rui Pan are co-first authors and contributed equally to this work.}

\author*[2]{\fnm{Yan} \sur{Zhang}}\email{zhangyan\_elyssa@163.com}

\affil[1]{\orgdiv{School of Statistics and Mathematics}, \orgname{Central University of Finance and Economics}, \orgaddress{\city{Beijing}, \country{China}}}

\affil[2]{\orgdiv{School of Statistics and Data Science}, \orgname{Shanghai University of International Business and Economics}, \orgaddress{\city{Shanghai}, \country{China}}}



\abstract{Community detection in citation networks offers a powerful approach to understanding knowledge flow and identifying core research areas within academic disciplines. This study focuses on knowledge source discovery in statistics by analyzing a weighted bipartite journal citation network constructed from 16,119 articles published in eight core journals from 2001 to 2023. To capture the inherent asymmetry of citation behavior, we explicitly preserve the bipartite structure of the network, distinguishing between citing and cited journals. For this task, we propose Bi-SCORE (Bipartite Spectral Clustering on Ratios-of-Eigenvectors), a computationally efficient and initialization-free spectral method designed for community detection in weighted bipartite networks with degree heterogeneity. We establish rigorous theoretical guarantees for the performance of Bi-SCORE under the weighted bipartite degree-corrected stochastic block model. Furthermore, simulation studies demonstrate its robustness across varying levels of sparsity and degree heterogeneity, where it outperforms existing methods. When applied to the real-world citation network, Bi-SCORE uncovers a six-community structure corresponding to key research areas in statistics, including applied statistics, methodology, theory, computation, and econometrics. These findings provide valuable insights into the intricate citation patterns and knowledge flow among statistical journals.}




\maketitle

\section{Introduction}\label{sec1}

Journal citation networks are powerful tools for studying the structure and evolution of science knowledge \citep{heneberg2016excessive}. In these networks, nodes represent journals and edges represent citation relationships. Specifically, if a paper published in journal $A$ cites a paper published in journal $B$, there is a directed edge from $A$ to $B$. This results in a directed and weighted network, where the weight reflects the citation counts between journals. Journal citation networks have been widely applied in various disciplines, such as economics \citep{jochmans2018semiparametric, hardle2024introduction}, education \citep{wang2016mapping}, medicine \citep{fiorentino2011clinical, fujita2022using}, neuroscience \citep{dworkin2020extent, vaca2021unreflective}, technology management \citep{lee2015uncovering}, e-commerce \citep{gong2024knowledge}, demography \citep{liu2005mapping}, materials science \citep{aykol2019network} and others. Specifically, journal citation networks can be used to identify influential journals \citep{nerur2005assessing, peng2013network}, uncover knowledge dissemination patterns \citep{calero2012seed, peng2015assortative}, and measure inter-disciplinarity \citep{franceschet2012large, leydesdorff2018betweenness}.

In particular, journal citation networks play a crucial role in knowledge source discovery by revealing key journals, researchers, methods, and theories within a field or interdisciplinary research \citep{li2023research, zhang2024latent}. Existing techniques for identifying knowledge sources using journal citation networks fall into three main categories. The first is the network centrality analysis, which assesses the importance of journals by calculating measures such as degree, closeness, and betweenness centrality \citep{leydesdorff2016aggregated, chen2024monitoring}. The second is intermediary analysis, which traces citation paths to identify influential journals that serve as bridges \citep{leydesdorff2018discontinuities}. The third is community detection, which clusters journals to uncover topical or disciplinary structures \citep{carusi2019scientific, butt2021systematic, xu2023covariate}. For instance, \cite{gomez2016visualization} identify research clusters based on direct citations, co-citation, and bibliographic coupling to generate a scientific knowledge map in library and information science. Later, \cite{calma2019journal} perform a comprehensive bibliometric analysis of publications in the business field, employing a label propagation algorithm to identify research themes. However, to the best of our knowledge, although the journal citation network plays a crucial role, its application in the field of statistics is relatively rare. 

In this study, we focus on the field of statistics and collect articles from eight core journals (listed in Table \ref{table:journal}) that span 2001 to 2023. The data are retrieved from the Web of Science (WOS, {\it www.webofknowledge.com}), including 16,119 articles and their corresponding 179,654 cited references. Using this dataset, we construct a bipartite (i.e., two-mode) journal citation network according to the citation relationships, where one mode comprises the eight core journals and the other the cited journals. This bipartite construction explicitly distinguishes the citing and cited roles, thereby preserving the inherent asymmetry of the citation behavior. This representation is particularly suitable for the discovery of the source of knowledge, which can be obscured in traditional one-mode projections \citep{hayes2025co}. Although prior studies mainly rely on one-mode networks that treat all journals as structurally equivalent \citep{guo2009knowledge, choudhury2020mining}, the bipartite perspective offers a more accurate structural representation. This formulation choice forms one of our key contributions. Similar bipartite network formulations have been adopted in domains such as recommendation systems \citep{maier2022bipartite}, ecological networks \citep{pilosof2017multilayer}, and financial transaction analysis \citep{chen2025model}.

To reveal the structural patterns of citation behaviors between journals, community detection can be performed on such bipartite journal citation networks. Existing approaches to bipartite community detection generally fall into three categories. The first is the projection method. It simply converts the bipartite network into one-mode projections and then applies traditional community detection techniques to the one-mode networks. Although this approach is simple to implement, it often results in information loss in the bipartite structure \citep{barber2007modularity, newman2013spectral, arthur2020modularity}. The second is the heuristic method, which identifies communities using techniques such as matrix decomposition, motif-based clustering, or discrete optimization, without relying on explicit generative models. These methods emphasize computational efficiency and scalability and have been widely applied in computer science domains such as information retrieval, recommendation systems, and web mining \citep{zhang2015community, wang2021efficient}. The third is the model-based method, which assumes that the network is generated from an underlying probabilistic model \citep{peixoto2015inferring, yen2020community}. A widely adopted model in this category is the Bipartite Stochastic Block Model (BiSBM), which assumes that edge probabilities depend solely on the community memberships of the two types of nodes \citep{larremore2014efficiently}. However, the original BiSBM does not account for degree heterogeneity, which is an empirically common feature of real-world networks \citep{karrer2011stochastic}. To address this limitation, recent extensions, such as degree-corrected BiSBM, introduce node-specific parameters to accommodate degree variability, improving both empirical performance and theoretical guarantees \citep{jin2015fast, barucca2016disentangling, brault2020consistency, zhao2024variational}.

Despite these advances, efficient and accurate community detection in large-scale and degree heterogeneous bipartite networks remains a challenge \citep{calderer2021community}. Traditional inference techniques such as pseudo-likelihood estimation \citep{amini2013pseudo, wang2023fast} and variational methods \citep{liu2024variational, zhao2024variational} are often computationally intensive or sensitive to initialization \citep{zhao2017survey}. As a scalable and initialization-free alternative, spectral clustering relies on the leading eigenvectors of the adjacency or Laplacian matrices. Several recent studies have explored spectral methods for community detection under BiSBM assumptions. For instance, \cite{qing2023community} propose a distribution-free model and a degree-corrected extension for weighted bipartite networks. Then they introduce bipartite spectral clustering (BiSC) and normalized BiSC for community detection. These methods are further extended to overlapping bipartite networks in \citet{qing2024bipartite}. Although these approaches represent important progress in model-based community detection, their theoretical guarantees often depend on the boundedness of deviations between adjacency entries and their expectations, a condition that may not hold in many weighted networks. Moreover, although normalized BiSC attempts to address degree heterogeneity, it can perform poorly in sparse settings. The limitation can also be observed in normalized principal component analysis for unipartite spectral clustering \citep{jin2015fast}. 

To overcome these challenges, we propose Bipartite Spectral Clustering on Ratios-of-Eigenvectors (Bi-SCORE), a new spectral methodology designed for bipartite degree-corrected block models. Bi-SCORE generalizes the widely used SCORE framework \citep{jin2015fast} in two key directions: it extends SCORE from unipartite to bipartite networks and from unweighted to weighted networks, thereby largely broadening its applicability. Specifically, Bi-SCORE utilizes the leading singular vectors of the bipartite adjacency matrix and applies a ratio transformation to mitigate the influence of nodal degrees. Unlike iterative optimization-based methods, Bi-SCORE is initialization-free and computationally efficient. Theoretically, Bi-SCORE offers strong consistency guarantees under mild regularity conditions, even in sparse and highly heterogeneous bipartite networks. Practically, we demonstrate the effectiveness of Bi-SCORE by applying it to a bipartite journal citation network in the field of statistics. The results reveal a clear and interpretable six-community structure, where each community corresponds to a meaningful subfield of statistical research, and the citation patterns among core journals further validate the domain-specific focus and cross-community knowledge flow captured by Bi-SCORE.

The remainder of this paper is organized as follows. In Section 2, we construct the journal citation network and provide a descriptive analysis. Section 3 proposes the Bi-SCORE method and presents its theoretical properties. Extensive simulation studies are conducted in Section 4 to evaluate the finite sample performance of our newly proposed methodology. The empirical results of the real data are shown in Section 5. Conclusions and discussions are presented in Section 6. All technical proofs are left in the Appendix.

\section{Journal Citation Network}\label{sec2}

\subsection{Data and Descriptive Analysis}

This study collects publications from 2001-01-01 to 2023-12-31 in eight representative journals in the field of statistics from the WOS. The journal name, number of publications, and JCR category are listed in Table \ref{table:journal}. These eight journals cover the core aspects of statistical science, i.e., statistical methodology, mathematical theory, applied statistics, and computation. Specifically, \emph{Biometrika}, \emph{Journal of Business \& Economic Statistics} (\emph{JBES}), \emph{Journal of the American Statistical Association} (\emph{JASA}), and \emph{Journal of the Royal Statistical Society Series B--Statistical Methodology} (\emph{JRSS-B}) aim to publish high-quality articles on statistical methodology. For instance, \emph{JBES} focuses on methodological innovations in economics, business, and finance, including the development and improvement of statistical methods, the adaptation of methods from other fields, and advances in computation. Regarding mathematical theory, the \emph{Annals of Statistics} (\emph{AOS}) focuses on publishing research papers on the theory of mathematical statistics and applied/interdisciplinary statistics. In the aspect of applied statistics, the \emph{Annals of Applied Statistics} (\emph{AOAS}) aims to provide a timely and unified forum for all areas of applied statistics. Furthermore, the \emph{Journal of Computational and Graphical Statistics} (\emph{JCGS}) and \emph{Statistics and Computing} (\emph{SC}) specialize in computation, with \emph{JCGS} having a particular focus on the latest techniques of graphical methods in statistics and data analysis. We refer to these eight journals as the {\it core journals}. In addition, we thoroughly retrieve the basic information for each publication, including the title, authors, publisher, publication year, keywords, abstract, citation count, and reference list. A concrete example is provided in Table \ref{table:example}. Specifically, the citation count represents the number of times each article has been cited by other works indexed by the \emph{Web of Science} until 2023. This dataset differs primarily from the data previously collected by our team \citep{gao2023large} in that it includes more detailed information about the journals in which the references are published.

\begin{table}[hbp]
\centering
\caption{An example of the basic information for publications, including the title, authors, publisher, publication year, citation count, and reference list.}
\label{table:example}
\begin{tabularx}{\textwidth}{>{\raggedright\arraybackslash}p{2.6cm} X}
\toprule
\textbf{Variable}                & \textbf{Detail} \\
\midrule
\textbf{Title}                & Using Heteroscedasticity to Identify and Estimate Mismeasured and Endogenous Regressor Models \\
\textbf{Authors}              & Lewbel, A \\
\textbf{Publisher}            & Journal of Business \& Economic Statistics \\
\textbf{Publication Year}     & 2012 \\
\textbf{Keywords}             & Endogeneity, Heteroscedastic errors, Identification, Measurement error, Partly linear model, Simultaneous system \\
\textbf{Abstract}             & This article proposes a new method of obtaining identification in mismeasured regressor models, triangular systems, and simultaneous equation systems. The method may be used in applications where other sources of identification, such as instrumental variables or repeated measurements ... \\
\textbf{Citation Count}       & 1,315 \\
\textbf{Reference List}       & 
\begin{tabular}[t]{@{}p{\dimexpr\textwidth-3.5cm\relax}@{}}
- Efficient estimation of models with conditional moment restrictions containing unknown functions (Econometrica) \\  
- The evidence on credit constraints in post-secondary schooling (Economic Journal) \\
- Two-step GMM estimation of the errors-in-variables model using high-order moments (Econometric Theory) \\
... (and more) \\
\end{tabular} \\

\bottomrule
\end{tabularx}
\end{table}

To better understand the performance of the eight core journals, we define the following three metrics: average citation, average reference, and self-citation rate. Specifically, the average citation is defined as the total citation counts of all publications in a specific journal over the total number of publications in the same journal. This measurement reflects the influence of publications in a journal. A higher average citation indicates that the research outputs of the journal have greater academic value and impact \citep{marx2015causes}. The average reference is defined as the total reference counts of all publications in a specific journal over the total number of publications in the same journal. A higher average reference may suggest that publications in a journal rely more heavily on and integrate research from other fields \citep{abt2002relationship}. In addition, the self-citation rate is defined as the total self-citation counts of all publications in a specific journal over the total citation counts of all publications in the same journal, multiplied by 100\%. It measures the proportion of citation counts that a journal receives for its own previously published work. Research on the self-citation rate is important in the study of knowledge source discovery, as it often reflects structural differences among journals, such as journal size, publisher country, and subject area \citep{tacskin2021self}.

The results are shown in Table \ref{table:three descriptive indicator}. Among the eight core journals, \emph{JRSS-B} has the highest average citation at 113.62, highlighting its influence in the field of statistical methodology. In addition, \emph{AOS} and \emph{JASA} also have high average citations, at 79.84 and 65.51, respectively, indicating that these two journals hold a prominent position in statistical theory. Most journals fall within the 30 to 40 range in terms of average reference, suggesting a moderate level of reliance on external sources in their research. There is an obvious variation in the self-citation rate among the eight core journals. Specifically, \emph{AOS}, \emph{Biometrika}, and \emph{JASA} have a relatively higher self-citation rate, all exceeding 10\%, while \emph{AOAS} and \emph{SC} have a relatively lower self-citation rate, below 3\%. This difference may be attributed to the fact that \emph{AOS}, \emph{Biometrika}, and \emph{JASA} concentrate on theoretical and methodological aspects of statistics, and publications in these journals often build upon prior work in the same journal. As a result, publications in these journals are more likely to cite earlier research from the same journal. In contrast, \emph{AOAS} and \emph{SC} focus on applied statistics and statistical computing, fields that are more interdisciplinary in nature, thus requiring the citation of research from a broader range of disciplines. The second column of Table \ref{table:three descriptive indicator} represents the total number of publications in the eight core journals. It can be seen that \emph{JASA} and \emph{AOS} have higher volume of publication compared to other core journals, demonstrating stronger academic influence. These findings emphasize that the eight core journals have a high academic reputation and authoritative status in statistics and related fields.

\begin{table}[h]
\centering 
\caption{Number of publications, average citation, average reference, and self-citation rate for core journals from 2001 to 2023.}
\label{table:three descriptive indicator}
\begin{tabularx}{\textwidth}{@{}>{\raggedright\arraybackslash}p{2.2cm} >{\centering\arraybackslash}p{2.3cm} >{\centering\arraybackslash}p{2.3cm} >{\centering\arraybackslash}p{2.3cm} >{\centering\arraybackslash}X@{}}
\toprule
\textbf{Journal Name} & \textbf{Number of Publications} & \textbf{Average Citation} & \textbf{Average Reference} & \textbf{Self-Citation Rate} \\
\midrule
\textbf{AOAS} & 1,793 & 37.05 & 42.38 & 1.68\% \\
\textbf{AOS} & 2,627 & 79.84 & 36.19 & 16.24\% \\
\textbf{Biometrika} & 1,861 & 45.46 & 25.68 & 10.82\% \\
\textbf{JRSS-B} & 1,290 & 113.62 & 37.84 & 5.99\% \\
\textbf{JASA} & 3,758 & 65.51 & 37.96 & 10.17\% \\
\textbf{JCGS} & 1,679 & 34.27 & 33.76 & 3.36\% \\
\textbf{JBES} & 1,451 & 49.97 & 38.34 & 3.89\% \\
\textbf{SC} & 1,660 & 34.47 & 36.45 & 2.61\% \\
\bottomrule
\end{tabularx}
\end{table}

\subsection{The Journal Citation Network}
\label{sec:jcn}

To investigate the knowledge source of the eight core journals, we construct a bipartite journal citation network. Specifically, we denote the eight core journals as one set of nodes, indexed by $i = 1, \dots, 8$. Meanwhile, the journals cited by the core journals are denoted as the other set of nodes, indexed by $j = 1, \dots, m$, where $m$ is the total number of journals including the core journals. As a result, the adjacency matrix of the journal citation network can be defined as $A = (A_{ij})\in\mathbb{R}^{8\times m}$, where $A_{ij}$ is a non-negative integer representing the frequency with which the core journal $i$ cites papers published in journal $j$. To focus more on key sources of knowledge and ensure a more concentrated analysis, we retain only journals with $A_{ij} \geq 40$, resulting in $m = 334$. Figure \ref{fig:journal citation network} visualizes the journal citation network, where the core journals are colored orange. The thickness of the edges represents the citation counts between journals, and we only display about 25\% names of all journals randomly, for the purpose of image clarity. It shows a clear community structure around the core journals, implying the need for community detection. For instance, the journals surrounding \emph{JBES} are concentrated in the fields of economics and finance (e.g., \emph{Journal of Econometrics}, \emph{Quantitative Economics}), focusing primarily on statistical methods related to economic models. Similarly, the journals surrounding \emph{AOS} are primarily associated with mathematical theory (e.g., \emph{Annals of Mathematics}, \emph{Probability Surveys}), reflecting its emphasis on probability theory, statistical inference, and other aspects of mathematical statistics. The above analysis provides a preliminary view of the community structure of the journal citation network.

\begin{figure}[htbp]
\centering
\includegraphics[trim=0.5cm 5cm 0.5cm 5cm, clip, width=0.9\textwidth]{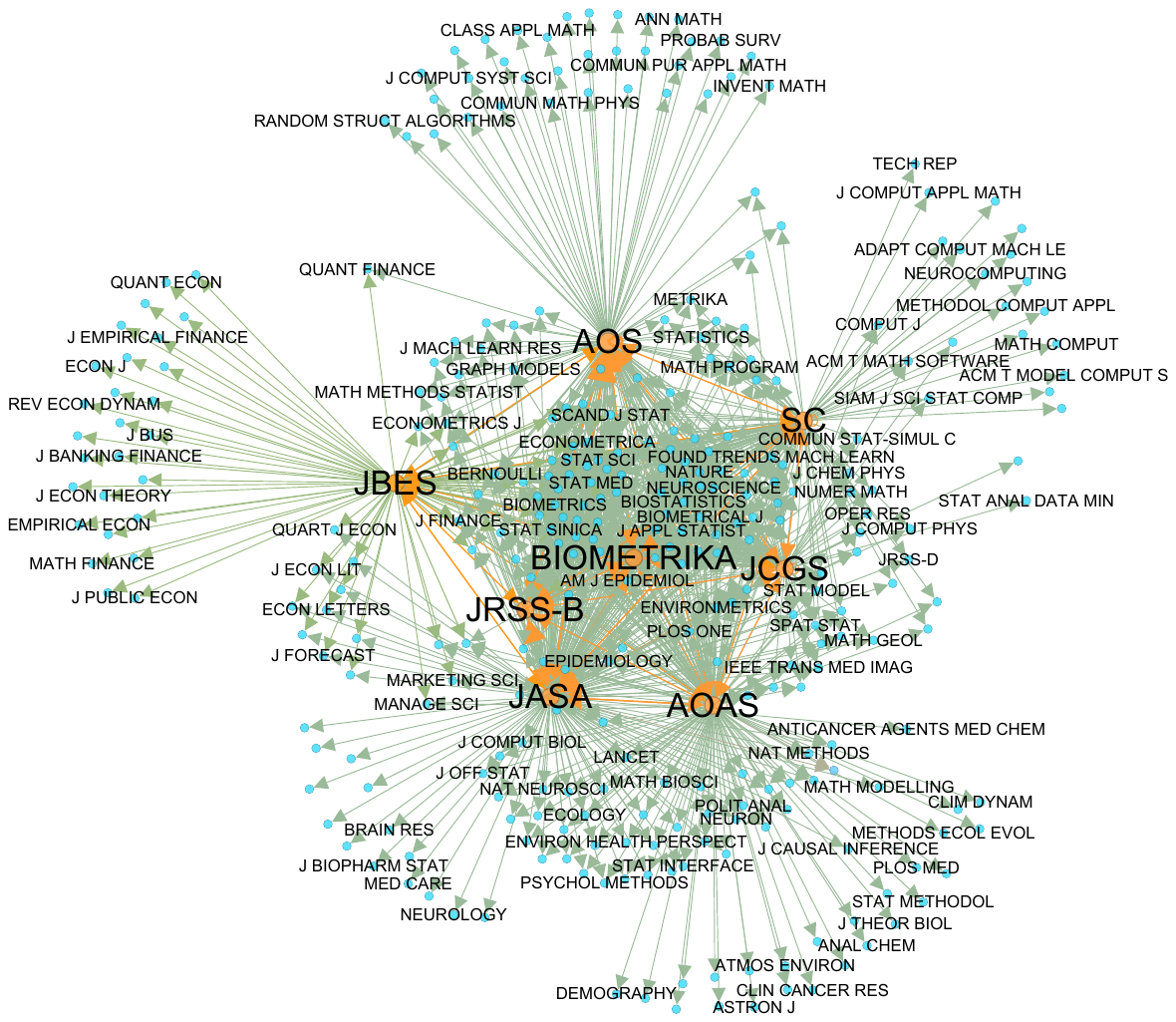}  
\caption{Journal Citation Network. It contains a total of 334 nodes and 1,063 edges, with the nodes for the eight selected journals highlighted in orange.}\label{fig:journal citation network}
\end{figure}

Table \ref{table:top 10 references} further lists the top 10 reference journals of the eight core journals. It can be seen that \emph{JASA}, \emph{AOS}, \emph{Biometrika} and \emph{JRSS-B} rank among almost the top five in terms of references in all core journals. This may be due to the high quality of their articles, making them an important knowledge source for other core journals during the observed period, as well as its comprehensive nature in statistical research and applications. In addition, Table \ref{table:top 10 references} reflects the reference preferences of the core journals, which is in line with the journal scope as presented in Section 2.1. 
For instance, \emph{AOAS} tends to cite articles related to applied statistics (e.g., \emph{Biometrics}, \emph{Statistics in Medicine}), while also referencing general science journals such as \emph{Proceedings of the National Academy of Sciences of the United States of America (PNAS)} and \emph{Nature}. This reflects its interdisciplinary orientation and broad scope of application.
\emph{AOS} primarily cites methodological and theoretical journals (e.g., \emph{Bernoulli}, \emph{IEEE Transactions on Information Theory (TIT)}, \emph{Annals of Mathematical Statistics (AMS)}), indicating its emphasis on foundational statistical and mathematical research.
\emph{Biometrika} tends to cite journals related to biostatistics and biomedical research (e.g., \emph{Biometrics}, \emph{Statistics in Medicine}). This indicates a concentrated emphasis on methodological development for biological and medical applications.
\emph{JRSS-B} and \emph{JASA} cite journals across multiple statistical domains, including biostatistics (e.g., \emph{Biometrics}), econometrics (e.g., \emph{Econometrica}), and machine learning (e.g., \emph{Journal of Machine Learning Research (JMLR)}). This reflects their integrative scope, including theoretical, applied, and computational statistics, as well as other interdisciplinary research.
\emph{JCGS} and \emph{SC} tend to cite journals in computational statistics, such as \emph{Computational Statistics \& Data Analysis (CSDA)} and \emph{JMLR}. This reflects their focus on statistical computing, data analysis, and the development of algorithmic methodologies.
\emph{JBES} cites more studies related to economics and finance (e.g., \emph{Journal of Econometrics (JOE)}, \emph{Econometrica}). This reflects its emphasis on the use of statistical methods to address problems in economic and financial research.

In addition to characterizing reference patterns of each eight core journals, we then shift our attention to the journals being cited, evaluating their impact and importance within the citation network. In-degree is widely used to measure this importance, reflecting how frequently a journal is cited by others. Specifically, we calculate the weighted in-degree for each cited journal $j$ as $C_w(j) = \sum_{i=1}^{8} A_{ij}$, where $A_{ij}$ is the weight of the edge from core journal $i$ to journal $j$ \citep{abbasi2011identifying}. The detailed definition of $A_{ij}$ is provided in Section \ref{sec3_1}.
Figure \ref{fig:weighted in-degree} presents the distribution of the weighted in-degree. The histogram indicates a highly right-skewed distribution, where the majority of nodes exhibit low-weighted in-degree (between 0 and 1,500), while a few nodes possess extremely high centrality values. Notably, \emph{AOS}, \emph{JASA}, \emph{JRSS-B}, and \emph{Biometrika} exhibit the highest weighted in-degree. This phenomenon is common in journal citation networks because the right-skewed distribution is often associated with the preferential attachment mechanism, in which newly published articles tend to cite already well-established, high-impact journals \citep{wang2008measuring}. Furthermore, structural centralization in academic disciplines reinforces this imbalance, making core journals crucial in academic dissemination and knowledge accumulation.

\begin{figure}[htbp]
\centering
\includegraphics[width=0.9\textwidth]{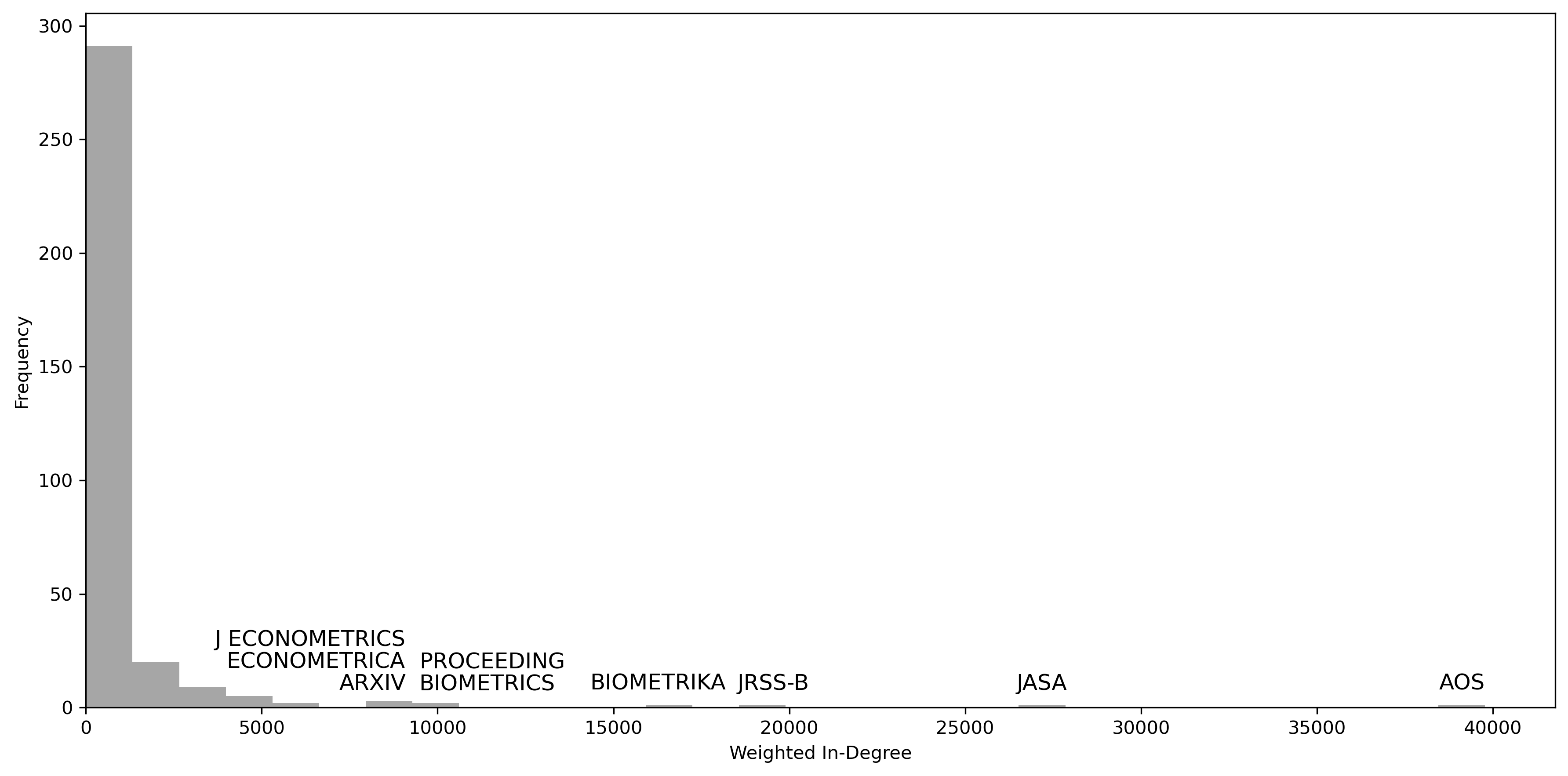}
\caption{Histogram of weighted in-degree. It has a highly right-skewed distribution and the weighted in-degree of most journals are between 0 and 1500.}\label{fig:weighted in-degree}
\end{figure}

\section{Methodology}\label{sec3}

In this section, we first introduce the Weighted Bipartite Degree-Corrected Stochastic Block Model (DCBM), which serves as the underlying generative model for the networks under study. Subsequently, we present Bi-SCORE, our novel spectral clustering algorithm designed to effectively identify community structures within these networks. Finally, we establish the theoretical properties of Bi-SCORE, providing rigorous guarantees on its performance.

\subsection{Weighted Bipartite DCBM}\label{sec3_1}
The stochastic block model (SBM) is a generative model designed for networks with community structures. In its most basic form, it applies to undirected, unweighted, and single-mode networks, where each node is assigned to one of the $K$ communities. The probability of forming an edge between two nodes is entirely dependent on their community memberships. However, the basic SBM ignores the inherent degree heterogeneity observed in real-world networks. To address this limitation, \cite{karrer2011stochastic} introduce the degree-corrected stochastic block model (DCBM), which accounts for degree heterogeneity, making it more suitable for analyzing empirical networks. Recently, \cite{zhao2024variational} extend the DCBM to weighted bipartite networks. Consider a network composed of $n$ row nodes, denoted by $\mathcal{V}^r = \{1, \dots, n\}$, and $m$ column nodes, denoted by $\mathcal{V}^c = \{1, \dots, m\}$. Let $A = (A_{ij}) \in \mR^{n \times m}$ be the adjacency matrix, where $A_{ij}$ denotes a non-negative integer weight between the $i$-th row node and the $j$-th column node. Assume that the $n$ row nodes are partitioned into $K$ communities, with their community labels given by $c^r = (c^r_1,\dots,c^r_n)^\top \in \{1,\dots,K\}^n$. Similarly, the $m$ column nodes are classified into $L$ communities, with labels $c^c = (c^c_1,\dots,c^c_m)^\top \in \{1,\dots,L\}^m$. To capture degree heterogeneity, we further incorporate two sets of parameters $\theta = (\theta_1,\dots, \theta_n)^\top \in \mR^n$ for row nodes and $\gamma = (\gamma_1,\dots,\gamma_m)^\top \in \mR^m$ for column nodes. We call them the {\it heterogeneity parameter}. These parameters model the individual propensities of nodes to form edges. Conditional on the community labels, the edge weights $A_{ij},\ i=1,\dots,n,\ j=1,\dots,m$ are assumed to be independent Poisson random variables with expected values defined as
\begin{align}
    \mathbb E(A_{ij}) = \theta_i \gamma_j B_{c^r_i c^c_j},
    \label{model1}
\end{align}
where $B= (B_{c^r_i c^c_j}) \in \mR^{K \times L}$ is a full rank matrix that quantifies the interaction strengths between row and column communities.
The community memberships can be represented by indicator matrices: $Z=(Z_{ik}) \in \mathbb{R}^{n \times K}$ for row nodes, where $Z_{ik}=1$ if $c^r_i = k$ and $Z_{ik}=0$ otherwise; and $W = (W_{jl}) \in \mathbb{R}^{m \times L}$ for column nodes, defined analogously. Throughout this paper, each row (column) node belongs to only one row (column) community. In other words, each row of $Z$ or $W$ has only one 1.  Additionally, we assume that $\text{rank}(Z) = K$ and $\text{rank}(W)=L$, which ensures that every row and column community is non-empty (i.e., contains at least one node).
In matrix form, the expected adjacency matrix $\Omega = \mathbb E(A)$ can be expressed as
\begin{align}
    \Omega = \mathbb E(A) = \Theta Z B W^\top \Gamma,
    \label{expected-A}
\end{align}
where $\Theta$ is an $n \times n$ diagonal matrix with $\Theta_{ii} = \theta_i$ and $\Gamma$ is an $m \times m$ diagonal matrix with $\Gamma_{jj} = \gamma_j$. 

\subsection{Algorithm}

Before introducing the algorithm, we first define some notation.
For a matrix $Q$, we use $Q_{i \cdot}$ to denote the $i$-th row of the matrix $Q$, and $Q_{\cdot j}$ to denote the $j$-th column of the matrix $Q$. In addition, we use $Q_{ij}$ to denote the $(i,j)$th entry in the matrix $Q$. For integer $0<j_1 \leq j_2$, let $Q_{\cdot j_1 \sim j_2}$ denote the matrix that is formed by extracting the $j_1$-th to $j_2$-th columns of the matrix $Q$, and $Q_{i j_1 \sim j_2}$ denote the vector that is formed by extracting the $i$-th row and the $j_1$-th to $j_2$-th columns of the matrix $Q$. For a vector $u$ and fixed $q > 0$, $\|u\|_q$ denote its $\ell_q$-norm. We remove the subscript if $q=2$. For a matrix $Q$, $\|Q\|$ denotes the spectral norm, and $\|Q\|_F$ denotes the Frobenius norm. We let $\|Q\|_{\min}$ denote the smallest singular value of the matrix $Q$. Let $\sigma_k(Q)$ denote the $k$-th largest singular value of matrix $Q$, and $\lambda_k(Q)$ denote the $k$-th largest eigenvalue of matrix $Q$ ordered by magnitude. 
For two positive sequences $\{a_n\}_{n=1}^\infty$ and $\{b_n\}_{n=1}^\infty$, we say $a_n \asymp b_n$ if there exists a positive constant $C$ such that $b_n / C \leq a_n \leq C b_n$ for sufficiently large $n$, i.e., $a_n$ and $b_n$ are of the same order. For a set $\mathcal{V}$, $|\mathcal{V}|$ denotes its cardinality.

In order to detect meaningful communities in weighted bipartite networks, we develop an algorithm called Bi-SCORE. Specifically, Bi-SCORE utilizes spectral clustering techniques for bipartite networks, based on singular value decomposition (SVD). It takes as input a bipartite adjacency matrix $A \in \mathbb{R}^{n \times m}$ and requires the number of clusters $K$ and $L$ to be pre-specified. The algorithm then outputs the community detection results for both the row and column nodes. The complete details of the Bi-SCORE algorithm are shown in Algorithm \ref{alg:Bi-SCORE}.

\begin{algorithm}
\caption{Bi-SCORE}
\label{alg:Bi-SCORE}
\begin{algorithmic}[1]
    \Require The bipartite adjacency matrix $A \in \mathbb{R}^{n \times m}$, the number of row clusters $K$, and the number of column clusters $L$.
    \Ensure The estimated community labels $\hat{c}^r$ for row nodes and $\hat{c}^c$ for column nodes.
    \State Perform SVD on $A$ to obtain the leading $\kappa = \min(K,L)$ left singular vector matrix $\widehat{U} \in \mathbb{R}^{n \times \kappa}$, and the leading $\kappa$ right singular vector matrix $\widehat{V} \in \mathbb{R}^{m \times \kappa}$.

    \State Set threshold values $\tau_n = \log (n)$ and $\tau_m = \log (m)$, then construct the matrices $\widehat{R}^r \in \mR^{n \times (\kappa-1)}$ and $\widehat{R}^c \in \mR^{m \times (\kappa-1)}$. For $i = 1,\dots,n$, $j=1,\dots,m$, and $k=1,\dots,(\kappa-1)$, define
    \begin{align}
    \label{hat-R}
        \widehat{R}^r_{ik} =
        \begin{cases}
            \tau_n, & \text{if } \frac{\widehat{U}_{i(k+1)}}{\widehat{U}_{i1}} > \tau_n \\
            \frac{\hat{U}_{k+1}(i)}{\hat{U}_1(i)}, & \text{if } \left\lvert \frac{\widehat{U}_{i(k+1)}}{\widehat{U}_{i1}} \right\rvert \leq \tau_n \\
            - \tau_n, & \text{if } \frac{\widehat{U}_{i(k+1)}}{\widehat{U}_{i1}} < -\tau_n
        \end{cases}
        , \quad
        \widehat{R}^c_{jk} =
        \begin{cases}
            \tau_m, & \text{if } \frac{\widehat{V}_{j(k+1)}}{\widehat{V}_{j1}} > \tau_m \\
            \frac{\hat{V}_{k+1}(i)}{\hat{V}_1(i)}, & \text{if } \left\lvert \frac{\widehat{V}_{j(k+1)}}{\widehat{V}_{j1}} \right\rvert \leq \tau_m. \\
            - \tau_m, & \text{if } \frac{\widehat{V}_{j(k+1)}}{\widehat{V}_{j1}} < -\tau_m
        \end{cases}
    \end{align}

    \State Perform $k$-means clustering on $\widehat{R}^r$ and $\widehat{R}^c$ by solving the following optimization problem:
    \begin{align*}
        \widehat{X}^r = \underset{X^r \in \mathcal{X}_{n,\kappa-1,K}}{\argmin} \left\| X^r - \widehat{R}^r \right\|_F^2, \quad \widehat{X}^c = \underset{X^c \in \mathcal{X}_{m,\kappa-1,L}}{\argmin} \left\| X^c - \widehat{R}^c \right\|_F^2,
    \end{align*}
    where $\mathcal{X}_{n,\kappa-1,K}$ represents the set of $n \times (\kappa-1)$ matrices with exactly $K$ distinct rows.
    
    \State Assign community memberships to row and column nodes based on the matrices $\widehat{X}^r$ and $\widehat{X}^c$.
\end{algorithmic}
\end{algorithm}

\begin{remark}
    The main goal of the Bi-SCORE algorithm is to find community labels for row nodes and column nodes. We do not directly estimate the heterogeneity parameters ($\theta$ and $\gamma$) or the community interaction strength matrix ($B$). This keeps the algorithm simpler and faster because we avoid estimating many parameters. If the estimation of $\theta$, $\gamma$, and $B$ is required, their values can be further estimated using methods such as moment estimation, after the community partitions have been obtained through the Bi-SCORE algorithm.
\end{remark}

\subsection{Theoretical Property}

Before obtaining the bounds of the misclustered nodes for Bi-SCORE, we first introduce some assumptions and propositions. We begin with an assumption regarding the matrix $B$.

\begin{assumption}
    \label{assump-B}
    The matrix $B$ satisfies
    \begin{gather*}
        B \text{ is non-singular},~
        0 \leq B_{kl} \leq 1 ~ \text{for} ~  k=1,\dots,K,~\text{and}~ l=1,\dots,L, \\
        BB^\top~\text{and}~B^\top B~\text{are non-negative and irreducible.}
    \end{gather*}
\end{assumption}
This assumption restricts the connection strengths $B_{kl}$ to lie within the interval $[0, 1]$, where the upper bound of 1 is primarily for notational simplicity. In practice, this condition can be relaxed as long as the entries of $B$ remain bounded. The irreducibility and non-negativity of $BB^\top$ and $B^\top B$ ensure that the leading left and right singular vectors of $B$ (associated with the largest singular value) are non-zero. This is crucial for the Bi-SCORE algorithm, as it guarantees non-vanishing denominators in spectral analysis, thereby enabling reliable community detection. This assumption is similar to Condition (2.7) of \cite{jin2015fast} and Assumption 1 of \cite{wang2020spectral}.

To introduce our assumptions about the heterogeneity parameters, we first define some notation. For each $k=1,\dots,K$ and $l=1,\dots,L$, define $\theta^{(k)} \in \mathbb{R}^n$ and $\gamma^{(l)} \in \mathbb{R}^m$ by
\begin{align*}
    \theta^{(k)}_i = \begin{cases}
    \theta_i & \text{if } c^r_i = k \\
    0 & \text{if } c^r_i \neq k 
    \end{cases}
    \quad \text{and} \quad
    \gamma^{(l)}_j = \begin{cases}
    \gamma_j & \text{if } c^c_j = l \\
    0 & \text{if } c^c_j \neq l.
    \end{cases}
\end{align*}
Next, define the following quantities: $\theta_{\min} \equiv \min_{1 \leq i \leq n} \theta_i$, $\theta_{\max} \equiv \max_{1 \leq i \leq n} \theta_i$, $\gamma_{\min} \equiv \min_{1 \leq j \leq m} \gamma_j$, and $\gamma_{\max} \equiv \max_{1 \leq j \leq m} \gamma_j$. We also define
\begin{align}
    \mathcal{Z} = \max(\theta_{\max}, \gamma_{\max}) \max(\|\theta\|_1, \|\gamma\|_1),
    \label{eq-Z}
\end{align}
a quantity that frequently appears in our analysis. In this work, we allow the heterogeneity parameters $\theta$ and $\gamma$ to scale with the network sizes $n$ and $m$.

\begin{assumption}
\label{assump-theta-gamma}
    The heterogeneity parameters $\theta$ and $\gamma$ satisfy
    \begin{gather}
        \theta_{\min} > 0,~ \gamma_{\min} >0, \label{theta-min} \\
        \|\theta^{(k_1)}\| \asymp \|\theta^{(k_2)}\|,~ \|\gamma^{(l_1)}\| \asymp \|\gamma^{(l_2)}\| \quad \text{for}~1 \leq k_1,k_2 \leq K,~\text{and}~1 \leq l_1,l_2 \leq L, \label{assump-eq-theta-gamma} \\
        \lim_{n,m \rightarrow \infty} \frac{\log(nm) \mathcal{Z} + \log^2(nm)}{\theta_{\min} \gamma_{\min} \|\theta\|_1 \|\gamma\|_1} = 0. \label{assump-eq-Z}
    \end{gather}
\end{assumption}
Condition \eqref{theta-min} ensures that every node has a strictly positive probability of forming connections. The requirement in \eqref{assump-eq-theta-gamma} enforces a form of balance between communities: the total connectivity within each row (or column) community is comparable to that of any other, preventing disproportionate dominance. Finally, condition \eqref{assump-eq-Z} implies that an accurate community recovery requires the total connection strength, represented by $\|\theta\|_1$ and $\|\gamma\|_1$, to grow faster than $\log(nm)$. Similar assumptions can be found in Conditions (2.6), (2.9), (2.13) of \cite{jin2015fast} and Assumption 2 of \cite{wang2020spectral}. We then present a few properties that follow directly from Assumption \ref{assump-theta-gamma}.
As a direct consequence of \eqref{assump-eq-theta-gamma}, we have: 
\begin{align}
\label{thetaktheta}
    \|\theta^{(k)}\| \asymp \|\theta\|~\text{and}~\|\gamma^{(l)}\| \asymp \|\gamma\|~\text{for}~1 \leq k \leq K~\text{and}~1 \leq l \leq L.
\end{align}
Moreover, since $\theta_{\min} \|\theta\|_1 \leq \|\theta\|^2$ and $\gamma_{\min} \|\gamma\|_1 \leq \|\gamma\|^2$, condition \eqref{assump-eq-Z} implies:
\begin{align}
\label{eq-lognmZ}
    \lim_{n,m \rightarrow \infty} \frac{\log(nm) \mathcal{Z} + \log^2(nm)}{\|\theta\|^2 \|\gamma\|^2} \leq \lim_{n,m \rightarrow \infty} \frac{\log(nm) \mathcal{Z} + \log^2(nm)}{\theta_{\min} \gamma_{\min} \|\theta\|_1 \|\gamma\|_1} = 0
\end{align}

Next, we analyze the singular vector matrix of the expected adjacency matrix $\Omega$ defined in \eqref{expected-A}, which captures the essential structure for clustering. We then show that the properties of $\Omega$ closely approximate those of the observed random adjacency matrix $A$. To proceed, we define $S \equiv \Psi_{\theta} B \Psi_{\gamma}^\top $, where $\Psi_{\theta}$ is a $K \times K$ diagonal matrix with entries $(\Psi_{\theta})_{kk} = \|\theta^{(k)}\| / \|\theta\|$, and $\Psi_{\gamma}$ is an $L \times L$ diagonal matrix with entries $(\Psi_{\gamma})_{ll} = \|\gamma^{(l)}\| / \|\gamma\|$. It captures the core community-level interactions scaled by the relative sizes of the heterogeneous parameters. 

\begin{proposition}
    \label{prop-UV}
    Let $\Omega = U \Lambda_{\Omega} V^\top$ be the compact singular value decomposition of the expected adjacency matrix $\Omega$. Then, the singular values of $\Omega$ are given by 
    \begin{align}
    \label{singularvalue}
        \sigma_k(\Omega) =
        \begin{cases} 
        \|\theta\|\|\gamma\|\sigma_k(S) & \text{if } 1 \leq k \leq \kappa, \\ 
        0 & \text{if } k > \kappa, 
        \end{cases}
    \end{align}
    where $S \equiv \Psi_{\theta} B \Psi_{\gamma}^\top $. Let $S = Y \Lambda_S H^\top$ denote the compact singular value decomposition of $S$. The singular vectors of $\Omega$ can be expressed in row form as 
    \begin{align}
    \label{rowform}
        U_{i \cdot} = \frac{\theta_i}{\|\theta^{(c^r_i)}\|} Y_{c^r_i \cdot} \quad \text{for}~1 \leq i \leq n \quad \text{and} \quad V_{j \cdot} = \frac{\gamma_j}{\|\gamma^{(c^c_j)}\|} H_{c^c_j \cdot} \quad \text{for}~1 \leq j \leq m,
    \end{align}
    and in column form as
    \begin{align}
    \label{columnform}
        U_{\cdot k} = \sum_{k^\prime=1}^K \frac{\theta^{(k^\prime)}}{\|\theta^{(k^\prime)}\|} Y_{k^\prime k} ~ \text{for}~1 \leq k \leq \kappa,~
        \text{and} ~ V_{\cdot l} = \sum_{l^\prime=1}^L \frac{\gamma^{(l^\prime)}}{\|\gamma^{(l^\prime)}\|} H_{l^\prime l} ~ \text{for}~1 \leq l \leq \kappa.
    \end{align}
    Furthermore, under Assumption \ref{assump-theta-gamma}, we have
    \begin{align}
    \label{UV}
        \|U_{i \cdot}\| \asymp \frac{\theta_i}{\|\theta\|} \quad \text{for}~ 1 \leq i \leq n, \quad \text{and} \quad \|V_{j \cdot}\| \asymp \frac{\gamma_j}{\|\gamma\|} \quad \text{for}~ 1 \leq j \leq m.
    \end{align}
\end{proposition}
The proof of Proposition \ref{prop-UV} can be found in Appendix \ref{proof-prop-UV}. Proposition \ref{prop-UV} provides a precise characterization of the spectral structure of the expected adjacency matrix $\Omega$. Specifically, it states that $\Omega$ has rank $\kappa$, where $\kappa$ is the rank of the matrix $B$. As a result, $\Omega$ admits a compact singular value decomposition $\Omega = U \Lambda_{\Omega} V^\top$, where $U \in \mathbb{R}^{n \times \kappa}$, $V \in \mathbb{R}^{m \times \kappa}$, and $\Lambda_{\Omega} \in \mathbb{R}^{\kappa \times \kappa}$ contains the nonzero singular values.  
Moreover, the singular vectors of $\Omega$ inherit a structured form that aligns with the community memberships. In particular, each row of $U$ (respectively, $V$) lies in the direction of the left (respectively, right) singular vector of the corresponding community from the decomposition of $S$, scaled by the heterogeneity parameter $\theta_i$ (or $\gamma_j$). This implies that nodes within the same community have vector representations that are parallel but differ in magnitude. The column-wise form of the singular vectors further reveals that each column of $U$ and $V$ can be expressed as a linear combination of community-specific vectors, weighted by normalized heterogeneity. Finally, \eqref{UV} illustrates that the spectral embedding magnitudes are governed by the connection tendencies of nodes.

The next proposition establishes bounds on the deviation between the singular vectors of the random adjacency matrix $A$ and those of its expectation $\Omega$.

\begin{proposition}
    \label{prop-U-UC}
    Let $\widehat{U}_{\cdot 1},\dots,\widehat{U}_{\cdot \kappa}$ and $\widehat{V}_{\cdot 1},\dots,\widehat{V}_{\cdot \kappa}$ denote the top $\kappa$ left and right singular vectors of $A$, respectively, and let $U_{\cdot 1},\dots,U_{\cdot \kappa}$ and $V_{\cdot 1},\dots,V_{\cdot \kappa}$ be the corresponding singular vectors of $\Omega$. Under Assumptions \ref{assump-B} and \ref{assump-theta-gamma}, there exist constants $C_U$ and $C_V$ with absolute value 1, as well as two orthogonal $(\kappa-1) \times (\kappa-1)$ matrices $O_U$ and $O_V$, such that for sufficiently large $n$ and $m$, with probability at least $1-1/n - 1/m$, the following bounds hold
    \begin{align*}
        &\|\widehat{U}_{\cdot 1} - U_{\cdot 1} C_U\| \leq C \frac{\sqrt{\log(nm) \mathcal{Z}} + \log(nm)}{\|\theta\| \|\gamma\|}, \\
        &\|\widehat{U}_{\cdot 2 \sim \kappa} - U_{\cdot 2 \sim \kappa} O_U\|_F  \leq C \frac{\sqrt{\log(nm) \mathcal{Z}} + \log(nm)}{\|\theta\| \|\gamma\|},\\
        &\|\widehat{V}_{\cdot 1} - V_{\cdot 1} C_V\| \leq C \frac{\sqrt{\log(nm) \mathcal{Z}} + \log(nm)}{\|\theta\| \|\gamma\|}, \\
        &\|\widehat{V}_{\cdot 2 \sim \kappa} - V_{\cdot 2 \sim \kappa} O_V\|_F \leq C \frac{\sqrt{\log(nm) \mathcal{Z}} + \log(nm)}{\|\theta\| \|\gamma\|},
    \end{align*}
    where $\mathcal{Z}$ is defined in \eqref{eq-Z}.
\end{proposition}
The proof of Proposition \ref{prop-U-UC} can be found in Appendix \ref{proof-prop-U-UC}. Proposition \ref{prop-U-UC} quantifies how closely the leading singular vectors of the observed adjacency matrix $A$ approximate those of its expectation $\Omega$. Specifically, the leading singular vectors (the first columns of $U$ and $V$) are determined only up to a sign, which is captured by the constants $C_U$ and $C_V$ of modulus one. For the remaining $\kappa - 1$ singular vectors, the ambiguity is up to an orthogonal rotation, captured by the matrices $O_U$ and $O_V$. 

Based on the expression in \eqref{rowform} and the approximation result from Proposition \ref{prop-U-UC}, we are now ready to explain the rationale behind removing the influence of node heterogeneity from the singular vectors in Algorithm \ref{alg:Bi-SCORE}. The core idea is to normalize each row of the singular vector matrix by its first entry. Specifically, for each $i=1,\dots,n$ and $j=1,\dots,m$, we define the ratio matrices $R^r$ and $R^c$ as 
\begin{align*}
    R^r_{i \cdot} = \frac{(U_{\cdot 2 \sim \kappa} O_{U})_{i \cdot}}{C_U U_{i1}} \quad \text{and} \quad R^c_{j \cdot} = \frac{(V_{\cdot 2 \sim \kappa} O_{V})_{j \cdot}}{C_V V_{j1}}.
\end{align*}
In other words, each row of $U$ is divided by its first entry, and the second through $\kappa$th columns are collected to form the ratio matrix $R^r$. The construction of $R^c$ follows analogously for the right singular vectors. To see the effect of this transformation, observe that
\begin{align}
    R^r_{i \cdot} = \frac{(U_{\cdot 2 \sim \kappa} O_{U})_{i \cdot}}{C_U U_{i1}} = \frac{U_{i 2 \sim \kappa} O_{U}}{C_U U_{i1}} = \frac{\frac{\theta_i}{\|\theta^{(c^r_i)}\|} Y_{c^r_i 2 \sim \kappa} O_{U}}{C_U \frac{\theta_i}{\|\theta^{(c^r_i)}\|} Y_{c^r_i 1}} = \frac{ Y_{c^r_i 2 \sim \kappa} O_{U}}{C_U Y_{c^r_i 1}},
    \label{Rr-transform}
\end{align}
where the third step follows from \eqref{columnform}. As shown in \eqref{Rr-transform}, the ratio eliminates the heterogeneity parameter $\theta_i$, and the resulting row vector in $R^r$ depends only on the community label $c_i^r$. Therefore, all nodes belonging to the same row community share the same row in $R^r$. A similar argument holds for $R^c$ and the column communities. This analysis highlights the key role of the ratio transformation in Algorithm \ref{alg:Bi-SCORE}.
To understand how the ratio matrices enable clear separation between communities, we present the following proposition.

\begin{proposition}
    \label{proposition R-R}    
    Let $R^r$ and $R^c$ be the ratio matrices constructed from the singular vectors of the expected adjacency matrix $\Omega$. Then, under Assumptions \ref{assump-B} and \ref{assump-theta-gamma}, for all $1 \leq i_1 < i_2 \leq n$ and $1 \leq j_1 < j_2 \leq m$, the following properties hold:
    \begin{gather*}
        \| R^r_{i_1 \cdot} - R^r_{i_2 \cdot}\| \geq 2, \quad \text{if} \quad c^r_{i_1} \neq c^r_{i_2}, \quad \text{and} \quad \| R^r_{i_1 \cdot} - R^r_{i_2 \cdot}\| =0, \quad \text{if} \quad c^r_{i_1} = c^r_{i_2}; \\
        \| R^c_{j_1 \cdot} - R^c_{j_2 \cdot}\| \geq 2, \quad \text{if} \quad c^c_{j_1} \neq c^c_{j_2}, \quad \text{and} \quad \| R^c_{j_1 \cdot} - R^c_{j_2 \cdot}\| =0, \quad \text{if} \quad c^c_{j_1} = c^c_{j_2}.
    \end{gather*}
\end{proposition}
The proof of Proposition \ref{proposition R-R} can be found in Appendix \ref{proof-proposition R-R}. Proposition \ref{proposition R-R} indicates that if two nodes $i_1$ and $i_2$ (or $j_1$ and $j_2$) belong to different communities (i.e., $c^r_{i_1} \neq c^r_{i_2}$ or $c^c_{j_1} \neq c^c_{j_2}$), then their rows in $R^r$ or $R^c$ will be sufficiently distinct. On the other hand, if two nodes belong to the same community (i.e., $c^r_{i_1} = c^r_{i_2}$ or $c^c_{j_1} = c^c_{j_2}$), their rows will be identical. This result suggests that the ratio matrices effectively separate nodes from different communities while clustering those from the same community together. Therefore, the structure captured in the ratio matrices $R^r$ and $R^c$ directly reflects the underlying community structure, which is crucial for the community detection task in Algorithm \ref{alg:Bi-SCORE}.
To assess how well the empirical ratio matrices approximate their population counterparts, we present the following result.

\begin{proposition}
    \label{R-R-F}
    For sufficient large $n$ and $m$, and under Assumptions \ref{assump-B} and \ref{assump-theta-gamma}, with probability at least $1-1/n-1/m$, the following inequalities hold:
    \begin{align*}
        \|\widehat{R}^r - R^r\|_F^2 \leq \frac{C \tau_n^2 \left\{\sqrt{\log(nm) \mathcal{Z}} + \log(nm)\right\}^2}{\theta_{\min}^2 \|\gamma\|^2}, \\ \|\widehat{R}^c - R^c\|_F^2 \leq \frac{C \tau_m^2 \left\{\sqrt{\log(nm) \mathcal{Z}} + \log(nm)\right\}^2}{\gamma_{\min}^2 \|\theta\|^2},
    \end{align*}
    where $\tau_n$ and $\tau_m$ are threshold values defined in Algorithm \ref{alg:Bi-SCORE}.
\end{proposition}
The proof of Proposition \ref{R-R-F} can be found in Appendix \ref{proof-R-R-F}.
In the following proposition, we provide a bound on the difference between the empirical clustering matrices $\widehat{X}^r$ and $\widehat{X}^c$ (derived from Algorithm \ref{alg:Bi-SCORE}) and $R^r$ and $R^c$. Specifically, $\widehat{X}^r$ (or $\widehat{X}^c$) is the matrix with exactly $K$ (or $L$) different rows (or columns), and it is the one closest to the ratio matrix $\widehat{R}^r$ (or $\widehat{R}^c$) in terms of the Frobenius norm.

\begin{proposition}
    \label{X-R}
    For sufficiently large $n$ and $m$, under Assumptions \ref{assump-B} and \ref{assump-theta-gamma}, with probability at least $1-1/n-1/m$, the following inequalities hold:
    \begin{align*}
        \|\widehat{X}^r - R^r\|_F^2 \leq \frac{C \tau_n^2 \left\{\sqrt{\log(nm) \mathcal{Z}} + \log(nm)\right\}^2}{\theta_{\min}^2 \|\gamma\|^2}, \\
        \|\widehat{X}^c - R^c\|_F^2 \leq \frac{C \tau_m^2 \left\{\sqrt{\log(nm) \mathcal{Z}} + \log(nm)\right\}^2}{\gamma_{\min}^2 \|\theta\|^2}.
    \end{align*}
\end{proposition}
The proof of Proposition \ref{X-R} can be found in Appendix \ref{proof-X-R}. Proposition \ref{X-R} suggests that the empirical clustering matrices will closely match the true community structure, which is essential for effective clustering in the Bi-SCORE algorithm.

Before presenting the main theorem for the Bi-SCORE algorithm, we first introduce some notation.
Recall that $\mathcal{V}^r$ and $\mathcal{V}^c$ are the sets of all row and column nodes, respectively. We define the sets of ``well-behaved'' nodes as follows: $\mathcal{W}^r \equiv \{1 \leq i \leq n: \|\widehat{X}_{i\cdot}^r - R_{i\cdot}^r \| \leq 1/2\}$ and $\mathcal{W}^c \equiv \{1 \leq j \leq m: \|\widehat{X}_{j\cdot}^c - R_{j\cdot}^c \| \leq 1/2\}$. This specific threshold of $1/2$ is chosen to ensure that the estimated values $\widehat{X}^r$ and $\widehat{X}^c$ for nodes from different communities remain sufficiently separated, which is crucial for accurate clustering. While any constant within the interval $(0,1)$ could serve as a valid threshold, we adopt $1/2$ in line with the practice established in \cite{wang2020spectral}.
By this definition, the sets $\mathcal{V}^r \setminus \mathcal{W}^r$ and $\mathcal{V}^c \setminus \mathcal{W}^c$ thus consist of nodes that are ``ill-behaved''. Let $n_k$ denote the number of row nodes in community $k$, for $k = 1, \dots, K$, and $m_l$ denote the number of column nodes in community $l$, for $l = 1, \dots, L$. The following theorem provides upper bounds on the number of misclustered row and column nodes under mild conditions.

\begin{theorem}
\label{theorem1}
    Under Assumptions \ref{assump-B} and \ref{assump-theta-gamma}, suppose that $\lvert \mathcal{V}^r \setminus \mathcal{W}^r\rvert < \min\{n_1,\dots,n_K\}$ and $\lvert \mathcal{V}^c \setminus \mathcal{W}^c\rvert < \min\{m_1,\dots,m_L\}$. Then, all nodes in $\mathcal{W}^r$ and $\mathcal{W}^c$ are correctly clustered by the Bi-SCORE algorithm. Moreover, for $n$ and $m$ large enough, with probability at least $1-1/n-1/m$,
    \begin{align*}
        \lvert \mathcal{V}^r \setminus \mathcal{W}^r \rvert < \frac{C \tau_n^2 \left\{\sqrt{\log(nm) \mathcal{Z}} + \log(nm)\right\}^2}{\theta_{\min}^2 \|\gamma\|^2}, \\
        \lvert \mathcal{V}^c \setminus \mathcal{W}^c \rvert < \frac{C \tau_m^2 \left\{\sqrt{\log(nm) \mathcal{Z}} + \log(nm)\right\}^2}{\gamma_{\min}^2 \|\theta\|^2}.
    \end{align*}
\end{theorem}
The proof of Theorem \ref{theorem1} can be found in Appendix \ref{Proof_of_Theorem_1}. In Theorem \ref{theorem1}, the assumption that the number of misclustered nodes is less than the size of the smallest community guarantees that at least one node of each community is correctly clustered. A similar assumption is made in \cite{jin2015fast} and \cite{wang2020spectral}. In the special case where the row nodes coincide with the column nodes, the network is simplified to a one-mode graph. When $\mathcal{Z}$ is of the order larger than $\log(nm)$, the bound on the number of misclustered nodes aligns with the result presented in Theorem 1 of \cite{wang2020spectral}.
Next, consider a simple case where $\theta$ and $\gamma$ are bounded by constants, i.e., $0 < \alpha \leq \min(\theta_{\min}, \gamma_{\min}) \leq \max(\theta_{\max}, \gamma_{\max}) \leq \beta < \infty$ with $\alpha$ and $\beta$ being positive constants and $\alpha \leq \beta$. For sufficiently large $n$ and $m$, we have $\mathcal{Z} \asymp \max(n,m)$, which dominates $\sqrt{\log(nm)}$. In this scenario, it follows that:
$\mathcal{Z} / (\theta_{\min}^2 \|\gamma\|^2) \leq \beta^2\alpha^{-4}$ and $\mathcal{Z} / (\gamma_{\min}^2 \|\theta\|^2) \leq \beta^2\alpha^{-4}$. When $\tau_n = \log(n)$ and $\tau_m = \log(m)$ with $n \asymp m$, the misclustering rates satisfy: 
\begin{align*}
    \lim_{n,m \rightarrow \infty} \frac{\lvert \mathcal{V}^r \setminus \mathcal{W}^r \rvert}{n} < \lim_{n,m \rightarrow \infty} \frac{C \log^2 (n) \log(nm)}{n} = 0, \\
    \lim_{n,m \rightarrow \infty} \frac{\lvert \mathcal{V}^c \setminus \mathcal{W}^c \rvert}{m} < \lim_{n,m \rightarrow \infty} \frac{C \log^2 (m) \log(nm)}{m} = 0. 
\end{align*}

\section{Simulation}

\subsection{Model Setup}

To evaluate the finite sample performance of our newly proposed methodology, we conduct a number of simulation studies in this section. We adopt the weighted bipartite DCBM described in Section \ref{sec3_1} to generate bipartite networks. First, we assign each of the \(n\) row nodes and the \(m\) column nodes to one of the \(K\) and \(L\) communities, respectively. Specifically, we independently sample the community label of each row node from the discrete uniform distribution over \(\{1, \dots, K\}\), and similarly for each column node over \(\{1, \dots, L\}\). The resulting community memberships are recorded in matrices \(Z \in \mathbb{R}^{n \times K}\) and \(W \in \mathbb{R}^{m \times L}\). Second, we specify the community-level interaction matrix \(B \in \mathbb{R}^{K \times L}\), which controls the baseline connectivity strength between communities. The details of generating $B$ are described in each scenario.
Third, we incorporate two degree parameters as \(\theta_i = \sqrt{\rho} \cdot a_i\) for the row nodes and \(\gamma_j = \sqrt{\rho} \cdot b_j\) for the column nodes, where \(a_i, b_j \sim \mathrm{Uniform}(0.5, 1)\) with \(i = 1, \dots, n\) and \(j = 1, \dots, m\). The parameter \(\rho \in (0,1]\) controls the overall strength of degree heterogeneity in the network. Specifically, higher values of \(\rho\) lead to greater variability in nodal degrees, whereas smaller values result in more homogeneous degree distributions. As a result, the adjacency matrix \(A \in \mathbb{R}^{n \times m}\) is generated with independent Poisson random variables whose means are specified by equation~\eqref{model1}.

\subsubsection*{Scenario 1: Balanced Sample Size}

In this scenario, we investigate the clustering performance under balanced sample size. Specifically, the combinations \((n, m)\) are chosen as \((500, 525)\), \((1000, 1050)\), \((1500, 1575)\), \((2000, 2100)\), \((2500, 2625)\), and \((3000, 3150)\). In addition, the degree heterogeneity parameter is set to \(\rho = 0.2\), and the number of communities is set to \(K = 2\) and \(L = 3\). Finally, the interaction matrix \(B \in \mathbb{R}^{2 \times 3}\) is specified as
\[
B =
\begin{bmatrix}
1.0 & 0.1 & 0.2 \\
0.3 & 0.9 & 0.1
\end{bmatrix}.
\]

\subsubsection*{Scenario 2: Unbalanced Sample Size}

In this scenario, we evaluate the effect of an unbalanced sample size between the two types of nodes. In particular, the combinations \((n, m)\) are chosen as \((50, 1500)\), \((100, 3000)\), \((150, 4500)\), \((200, 6000)\), \((250, 7500)\), and \((300, 9000)\). In addition, we set the degree heterogeneity parameter to \(\rho = 0.9\), and the number of communities is chosen the same as that in Scenario 1. Finally, the interaction matrix \(B \in \mathbb{R}^{2 \times 3}\) is specified as
\[
B =
\begin{bmatrix}
0.9 & 0.1 & 0.2 \\
0.2 & 1.0 & 0.1
\end{bmatrix}.
\]
This setting mimics some bipartite networks in the real world, such as in academic collaboration networks, where one type of node is typically much smaller in number than the other \citep{newman2001structure, zhou2007co}. Specifically, the journal citation network constructed in our study also exhibits the same phenomenon.

\subsubsection*{Scenario 3: Degree Heterogeneity}

In this scenario, we examine the influence of degree heterogeneity on clustering performance. We fix the number of row and column nodes at \((n, m)\) = \((1000, 1050)\). In addition, the degree heterogeneity parameter \(\rho\) is varied in \(\{0.1, 0.2, \ldots, 1.0\}\). As \(\rho\) increases, nodal degrees become more heterogeneous, raising greater challenges for community detection \citep{dall2019revisiting, kojaku2024network}. Finally, the number of communities remains the same as in Scenario 1, and the interaction matrix \(B\) is specified as
\[
B =
\begin{bmatrix}
1.0 & 0.2 & 0.1 \\
0.2 & 0.9 & 0.3
\end{bmatrix}.
\]

\subsubsection*{Scenario 4: Number of Communities}

In this scenario, we investigate the effect of varying the number of communities on clustering performance. We fix the number of nodes in the row and column at \((n, m) = (1000, 1050)\), and set the degree heterogeneity parameter to \(\rho = 0.9\). Specifically, we consider three combinations of community numbers \((K = 2, L = 3)\), \((K = 3, L = 5)\), and \((K = 3, L = 7)\). The corresponding community interaction matrices \(B_1\), \(B_2\), and \(B_3\) are given as 
\[
B_1 =
\begin{bmatrix}
    1.0 & 0.3 & 0.1 \\
    0.3 & 1.0 & 0.1
\end{bmatrix},
B_2 =
\begin{bmatrix}
    1.0 & 0.2 & 0.3 & 0.1 & 0.3 \\
    0.3 & 0.1 & 1.0 & 0.1 & 0.3 \\
    0.3 & 0.1 & 0.3 & 0.2 & 1.0
\end{bmatrix},
B_3 =
\begin{bmatrix}
    1.0 & 0.2 & 0.1 & 0.3 & 0.1 & 0.3 & 0.1 \\
    0.3 & 0.1 & 0.3 & 1.0 & 0.2 & 0.3 & 0.1 \\
    0.3 & 0.1 & 0.1 & 0.3 & 0.1 & 1.0 & 0.3
\end{bmatrix}.
\]



\subsection{Performance Measurement and Results}

To assess the performance of different clustering methods, we conduct simulation experiments with \(T = 1{,}000\)  replications for each scenario. In the \(t\)th replication, the clustering performance is evaluated using the error rate (ErrorRate) and the adjusted Rand index (ARI) \citep{hubert1985comparing}, which are defined as follows.
\begin{align*}
    \mathrm{ErrorRate}^{(t)} = \max \left\{
        \frac{1}{n} \sum_{i=1}^n \mathbb{I} \left( \hat{c}^{r(t)}_i \ne c^r_i \right),
        \frac{1}{m} \sum_{j=1}^m \mathbb{I} \left( \hat{c}^{c(t)}_j \ne c^c_j \right)
    \right\},
\end{align*}

\begin{align*}
\mathrm{ARI}^{(t)} = \min \left\{ \mathrm{ARI}^{r(t)},\, \mathrm{ARI}^{c(t)} \right\},
\end{align*}

\noindent where
\begin{align*}
\mathrm{ARI}^{r(t)} = 
\frac{
    \sum_{k=1}^{K} \sum_{l=1}^{K} \binom{n_{kl}^{r(t)}}{2} - 
    \Big[ \sum_{k=1}^{K} \binom{n_{k\cdot}^{r(t)}}{2} \sum_{l=1}^{K} \binom{n_{\cdot l}^{r(t)}}{2} \Big] \big/ \binom{n}{2}
}{
    \tfrac{1}{2} \Big[ \sum_{k=1}^{K} \binom{n_{k\cdot}^{r(t)}}{2} + \sum_{l=1}^{K} \binom{n_{\cdot l}^{r(t)}}{2} \Big] -
    \Big[ \sum_{k=1}^{K} \binom{n_{k\cdot}^{r(t)}}{2} \sum_{l=1}^{K} \binom{n_{\cdot l}^{r(t)}}{2} \Big] \big/ \binom{n}{2}
},
\end{align*}

\begin{align*}
\mathrm{ARI}^{c(t)} = 
\frac{
    \sum_{k=1}^{L} \sum_{l=1}^{L} \binom{n_{kl}^{c(t)}}{2} - 
    \Big[ \sum_{k=1}^{L} \binom{n_{k\cdot}^{c(t)}}{2} \sum_{l=1}^{L} \binom{n_{\cdot l}^{c(t)}}{2} \Big] \big/ \binom{m}{2}
}{
    \tfrac{1}{2} \Big[ \sum_{k=1}^{L} \binom{n_{k\cdot}^{c(t)}}{2} + \sum_{l=1}^{L} \binom{n_{\cdot l}^{c(t)}}{2} \Big] -
    \Big[ \sum_{k=1}^{L} \binom{n_{k\cdot}^{c(t)}}{2} \sum_{l=1}^{L} \binom{n_{\cdot l}^{c(t)}}{2} \Big] \big/ \binom{m}{2}
}.
\end{align*}
In the above equations, \(n_{kl}^{r(t)}\) denotes the number of row nodes that belong to cluster \(k\) in the true partition and are assigned to cluster \(l\) in the predicted partition in the \(t\)th replication, where \(1 \le k, l \le K\). The terms \(n_{k\cdot}^{r(t)} = \sum_l n_{kl}^{r(t)}\) and \(n_{\cdot l}^{r(t)} = \sum_k n_{kl}^{r(t)}\) represent the marginal totals of the row and column, respectively. The quantities \(n_{kl}^{c(t)}\), \(n_{k\cdot}^{c(t)}\), and \(n_{\cdot l}^{c(t)}\) for column nodes are defined analogously, with \(1 \le k, l \le L\).
Specifically, the adjusted Rand index is bounded above by \(1\) and equals \(0\) under random labeling, which higher values indicate better agreement \citep{vinh2010information}. The final error rate and adjusted Rand index are taken as the averages over the \(1{,}000\) replications.
For each setting, we first generate the bipartite network according to the model described in Section 3.1, and extract the giant component if necessary. Subsequently, we apply Bi-SCORE, nBiSC \citep{qing2023community}, and spectral clustering \citep{dhillon2001co}. The results are summarized in Figure \ref{fig:error_curves}.

\begin{figure}[htbp]
\centering
\includegraphics[width=0.9\textwidth, trim=0 0 0.5cm 0, clip]{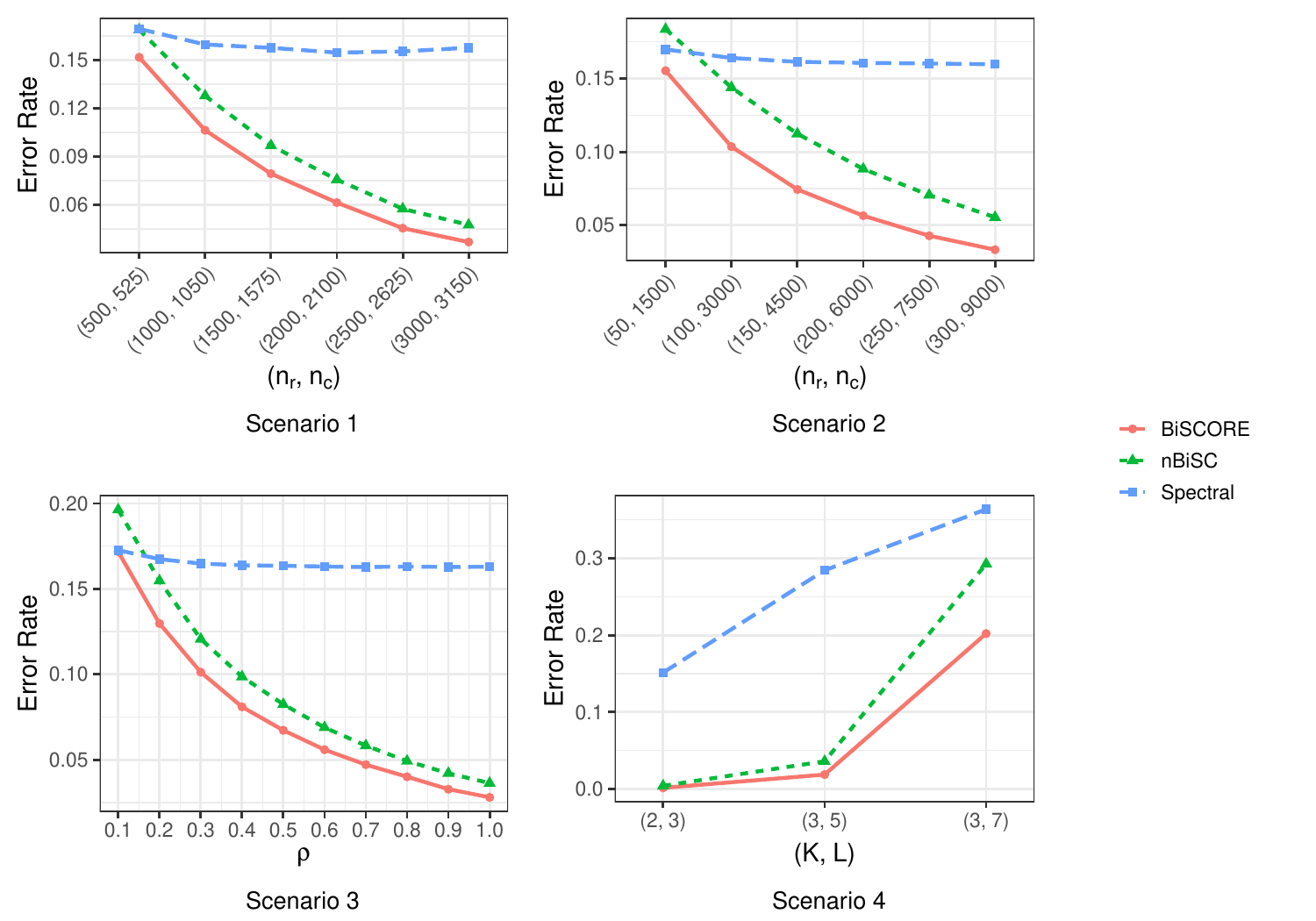}
\caption{
The error rate curves across four scenarios: 
(top-left) balanced sample size, 
(top-right) unbalanced sample size, 
(bottom-left) degree heterogeneity, 
and (bottom-right) network sparsity. 
Bi-SCORE, nBiSC, and spectral clustering methods are compared across all settings. Bi-SCORE achieves lower error rates under most scenarios, especially with increasing sample size, higher degree heterogeneity, or reduced network sparsity.
}
\label{fig:error_curves}
\end{figure}

Figure \ref{fig:error_curves} displays the error rate curves in four scenarios. In general, Bi-SCORE performs better than nBiSC and spectral clustering in terms of the error rate. We observe the following three phenomena. 
First, in both Scenario 1 and Scenario 2, the error rates of Bi-SCORE and nBiSC decrease notably as the sample size increases, while the performance of spectral clustering remains relatively unchanged. Specifically, Bi-SCORE achieves the lowest error rates and exhibits the fastest rate of improvement with increasing sample size.
Second, in Scenario 3, increasing \(\rho\) leads to greater degree heterogeneity in the network, which in turn improves the performance of both Bi-SCORE and nBiSC. In particular, when \(\rho > 0.2\), Bi-SCORE outperforms the other two methods, highlighting its adaptability to heterogeneous network structures.
Third, in Scenario 4, the performance of all methods declines as the network sparsity increases. Although the advantage of Bi-SCORE over nBiSC becomes less pronounced under this setting, Bi-SCORE still maintains a slight advantage, especially in less complex cases.

\begin{figure}[htbp]
\centering
\includegraphics[width=0.9\textwidth, trim=0 0 0.5cm 0, clip]{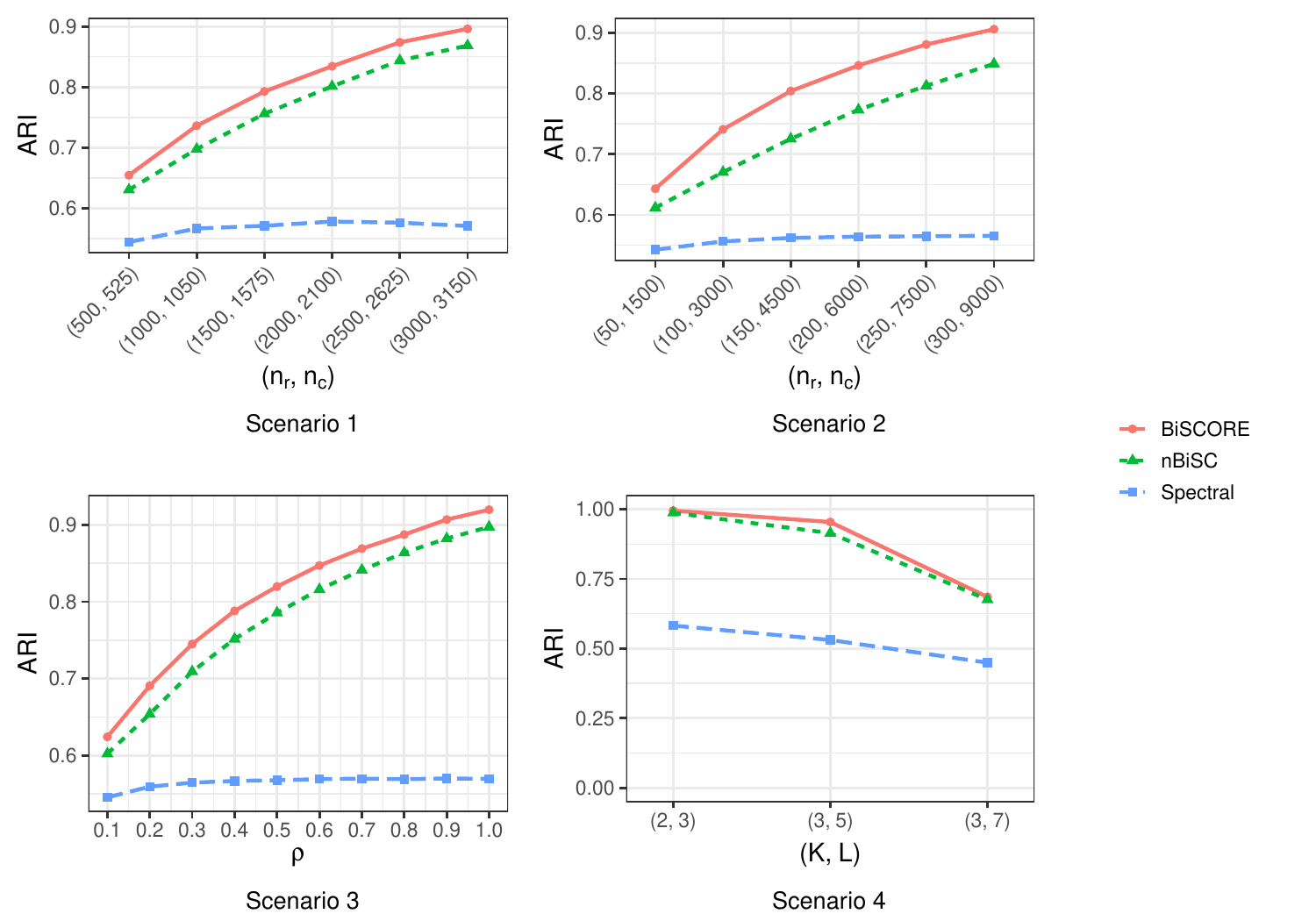}
\caption{
The adjusted Rand index curves across four scenarios: 
(top-left) balanced sample size, 
(top-right) unbalanced sample size, 
(bottom-left) degree heterogeneity, 
and (bottom-right) network sparsity. 
Bi-SCORE, nBiSC, and spectral clustering methods are compared across all settings. Bi-SCORE achieves higher adjusted Rand indices under most scenarios, especially with increasing sample size, higher degree heterogeneity, or reduced network sparsity.
}
\label{fig:ari_curves}
\end{figure}

Figure~\ref{fig:ari_curves} displays the adjusted Rand index curves in four scenarios. In general, Bi-SCORE achieves higher adjusted Rand indices than both nBiSC and spectral clustering. We observe the following three phenomena. 
First, in Scenario 1 and Scenario 2, the adjusted Rand indices of Bi-SCORE and nBiSC increase substantially as the sample size grows, while the adjusted Rand index of spectral clustering remains relatively lower. In particular, Bi-SCORE attains the highest adjusted Rand index and demonstrates the fastest rate of improvement with increasing sample size.
Second, in Scenario 3, both Bi-SCORE and nBiSC achieve higher adjusted Rand indices as the degree heterogeneity parameter \(\rho\) increases. Specifically, Bi-SCORE achieves the highest adjusted Rand index among all three methods.
Third, in Scenario 4, the adjusted Rand indices for all methods decrease as the network sparsity increases. Although the advantage of Bi-SCORE over nBiSC becomes less pronounced in this setting, Bi-SCORE still maintains a slight advantage, especially in less complex cases.

\section{Knowledge Discovery of Statistical Journals}\label{sec5}

\subsection{Knowledge Discovery of the Journal Citation Network}\label{sec5_1}

In this section, we apply the newly proposed Bi-SCORE to our journal citation network described in Section \ref{sec:jcn}. Since our goal is to conduct knowledge source discovery in the eight core statistical journals, we only need to identify the community structure among the column nodes, i.e., cited journals.
To this end, we first need to decide the number of communities \(L\). We apply Bi-SCORE for \(L = \{ 3, 4, 5, 6, 7, 8 \}\), and compare the clustering results for each \(L\). The results show that the community structure is clearest when \(L = 6\). Therefore, we set \(L = 6\) in the subsequent analysis. The result of the community structure is further visualized in Figure~\ref{fig:community_dectection}. It reveals that the journal citation network is divided into 6 communities, which are summarized in detail in the following.

\begin{figure}[htbp]
\centering
\includegraphics[trim=0.5cm 4cm 0.5cm 4cm, clip, width=0.9\textwidth]{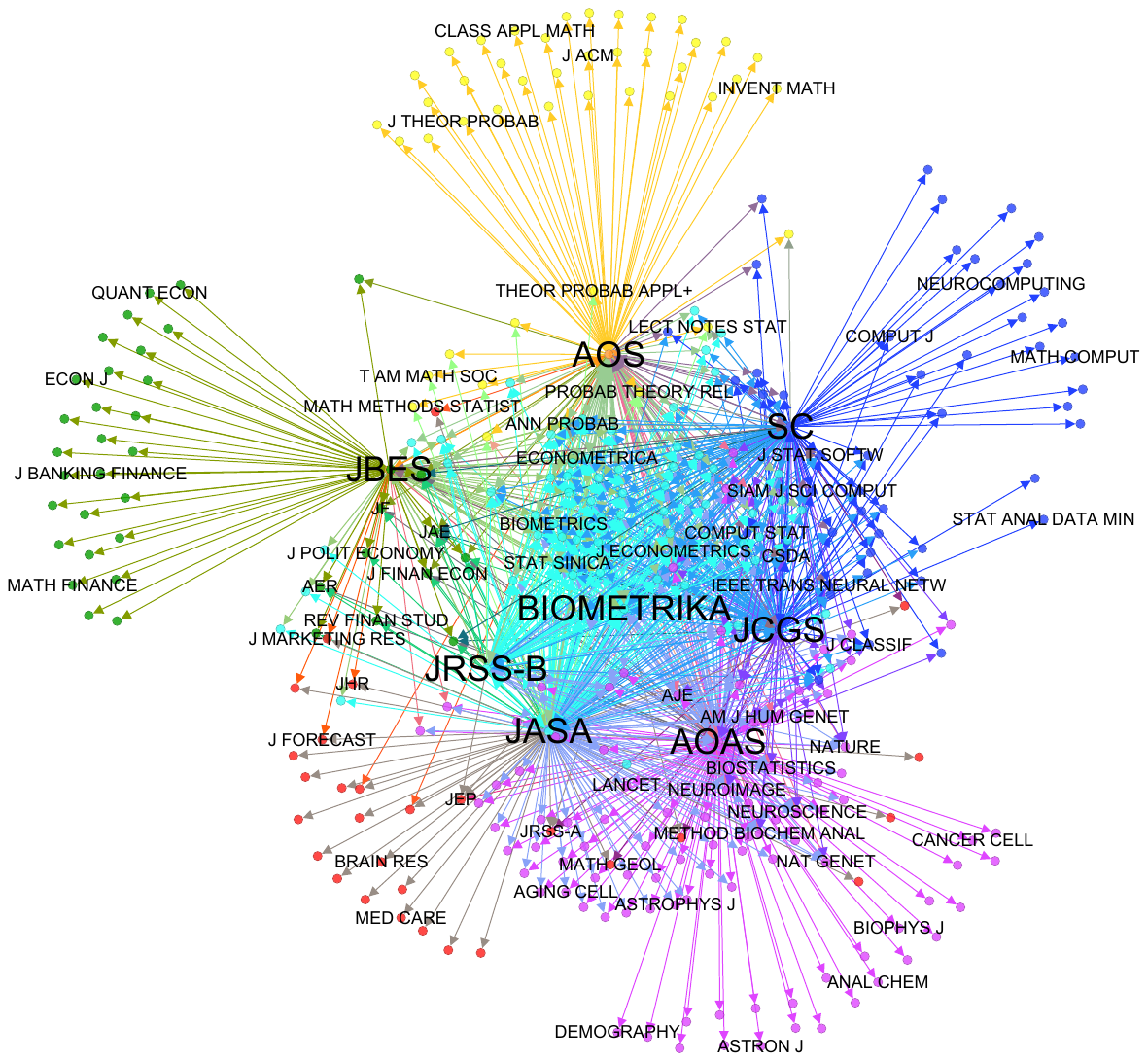}
\caption{Visualization of the Journal Citation Network with communities detected by Bi-SCORE. The 6 identified communities are visualized in different colors: Applied Statistics (purple), Statistical Methodology (cyan), Computational Statistics (blue), Mathematical Statistics (yellow), Econometrics and Business Statistics (dark green), and Others (red).}
\label{fig:community_dectection}
\end{figure}

\begin{itemize}

     \item \textbf{Applied Statistics (Purple)}. This community contains 98 journals, which primarily focus on the practical application of statistical methods. Research in this community emphasizes real-world data analysis, interdisciplinary collaboration, and methodological adaptation. The most representative journal in this community is \emph{AOAS}, as expected for a core journal. In addition, \emph{Bioinformatics}, \emph{Nature}, \emph{Neuroscience}, and \emph{American Journal of Epidemiology (AJE)} also have high weighted in-degree within this community.
    
    \item \textbf{Statistical Methodology (Cyan)}. This community contains 73 journals, which primarily focus on methodological development and innovation in statistical science. Research topics include new statistical methods, model selection, and statistical inference. The main journals in this community are \emph{Biometrika}, \emph{JASA}, and \emph{JRSS-B} (core journals). Furthermore, \emph{Biometrics}, \emph{Econometrica}, and \emph{Statistica Sinica} also exhibit high weighted in-degree in this community.
    
    \item \textbf{Computational Statistics (Blue)}. This community contains 50 journals, which emphasize computational methodology, statistical software, and data science techniques. These features highlight the central role of algorithmic and computational innovation in modern statistical research, bridging traditional statistical modeling with emerging data science applications. Specifically, \emph{JCGS}, \emph{SC}, \emph{CSDA}, and \emph{Journal of Statistical Software} have the highest weighted in-degree. 
    
    \item \textbf{Mathematical Statistics (Yellow)}. This community contains 45 journals, which focus primarily on mathematical analysis, probability theory, and statistical inference. Research in this community emphasizes theoretical innovation and foundational contributions. The leading journal in this community is \emph{AOS}, which is the core journal defined in our journal citation network. In addition, \emph{Annals of Probability}, \emph{Probability Theory and Related Fields}, and \emph{Mathematical Methods of Statistics} also have high weighted in-degree in this community.
    
    \item \textbf{Econometrics and Business Statistics (Dark Green)}. This community contains 39 journals, with a primary focus on econometrics, finance, and quantitative economics. These citation patterns reflect a close integration of econometric methods with empirical research in economics and finance. Notably, \emph{JBES}, \emph{Journal of Applied Econometrics (JAE)}, \emph{American Economic Review (AER)}, and \emph{Journal of Finance (JF)} have the highest weighted in-degree within this community.
    
    \item \textbf{Others (Red)}. This community contains 29 journals spanning a wide range of fields, including education, psychology, medicine, biology, economics, political science, engineering, and others. Research in this group emphasizes the application and advancement of statistical methods in emerging interdisciplinary fields. Representative journals in this community include \emph{Journal of Educational Psychology (JEP)}, \emph{Journal of Forecasting}, \emph{Journal of Human Resources (JHR)}, and \emph{Journal of the Royal Statistical Society Series A (JRSS-A)}.
    
\end{itemize}

Figure~\ref{fig:core_citation} reveals the knowledge sources for the core journals in the 6 identified communities. The following three main phenomena can be observed. 
First, all core journals extensively refer to the ``Statistical Methodology'' community. This suggests the central role of methodological development, which connects different fields of statistics from theory to application. In contrast, all core journals rarely cite from the ``Others" community, indicating a concentration of knowledge sources within statistics-related domains.
Second, both \emph{AOAS} and \emph{JASA} cite heavily from the ``Applied Statistics" community, reflecting their strong emphasis on empirical research and real-world applications. This may be due to the special section \emph{Journal of the American Statistical Association: Applications and Case Studies (JASA ACS)}, which is devoted to applied statistical analyses of real data. Specifically, \emph{JASA} also cites extensively from the ``Statistical Methodology'' community, highlighting its integrated role in both methodological development and applied research. 
Third, \emph{SC} and \emph{JCGS} cite frequently from the ``Computational Statistics" community. This reflects their emphasis on computational tools, graphical methods, and data-driven applications in diverse scientific domains. Moreover, \emph{AOS} cites heavily from the ``Mathematical Statistics" community, and \emph{JBES} cites heavily from the ``Econometrics and Business Statistics" community. This reveals the distinct orientations of the two journals: \emph{AOS} emphasizes theoretical innovation in probability and statistical inference, while \emph{JBES} focuses on the development and application of statistical methods in economics, finance, and business.

\begin{figure}[htbp]
\centering
\includegraphics[width=0.9\textwidth]{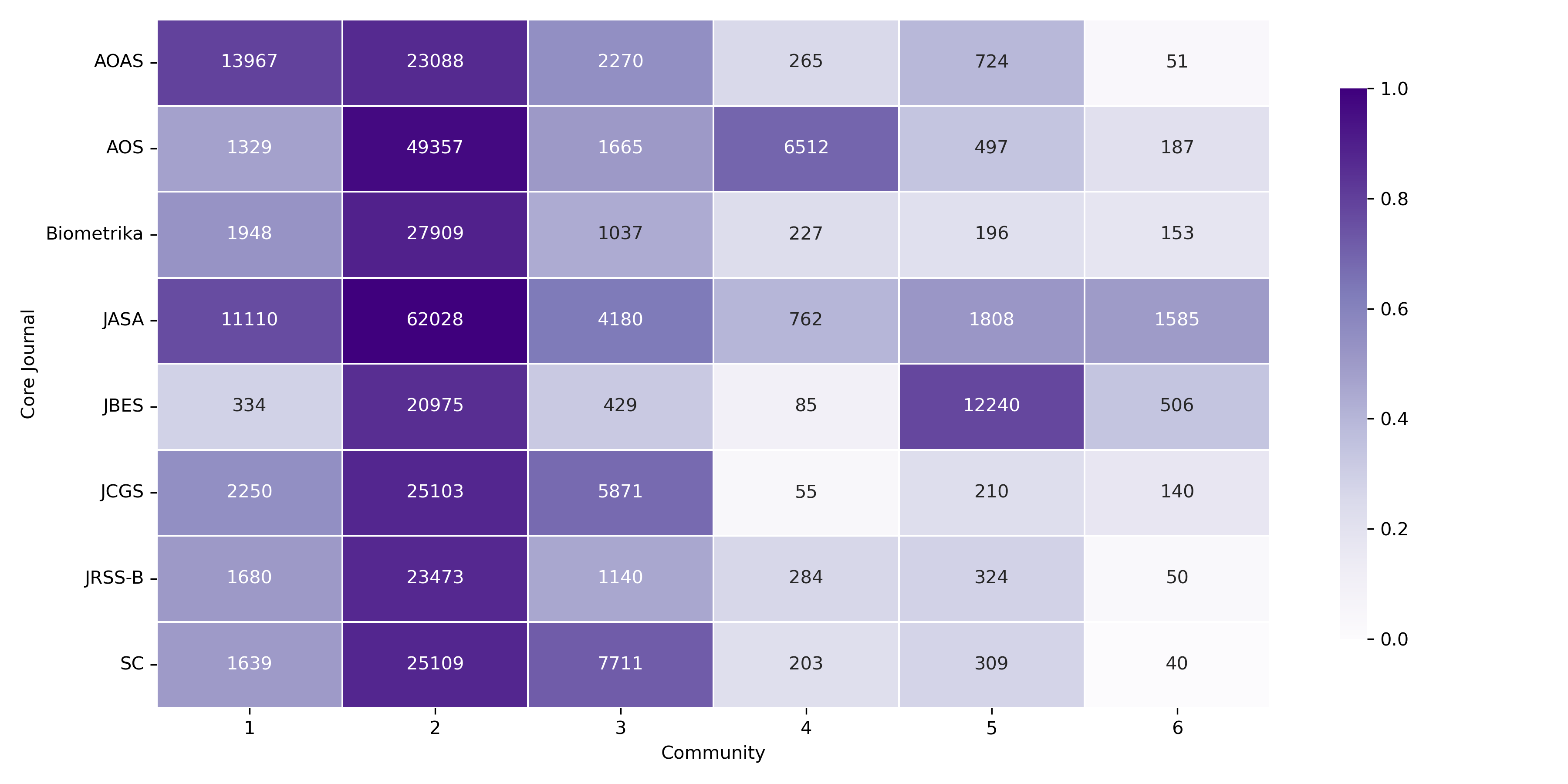}
\caption{Heatmap of knowledge sources for core journals across the 6 identified communities. The x-axis labels correspond to the 6 communities: (1) Applied Statistics, (2) Statistical Methodology, (3) Computational Statistics, (4) Mathematical Statistics, (5) Econometrics and Business Statistics, and (6) Others. Each cell indicates the number of journals from a given community cited by a specific core journal. Darker colors indicate higher citation counts (based on log-transformed and normalized values). The figure reveals distinctive citation patterns and research emphases of the core journals.}
\label{fig:core_citation}
\end{figure}

\subsection{Knowledge Discovery of the Applied Statistics Community}

In Section \ref{sec5_1}, we analyze the knowledge sources of our journal citation network and divide it into 6 communities. In addition, as the largest community in the journal citation network (98 journals), the ``Applied Statistics" community contains the core journal AOAS and many other representative journals. To further investigate the knowledge sources of this community, we apply Bi-SCORE to its subnetwork. Similarly, we apply Bi-SCORE for \(L = \{2, 3, 4, 5\}\) and compare the clustering results. Among these alternatives, the partition obtained with \(L = 4\) provides the most coherent and interpretable community division. Consequently, \(L = 4\) is adopted in the subsequent investigation. The result of the community structure is visualized in Figure~\ref{fig:sub_dectection}, and its characteristics are discussed in detail below.

\begin{figure}[htbp]
\centering
\includegraphics[trim=0.5cm 6cm 0.5cm 6cm, clip, width=0.9\textwidth]{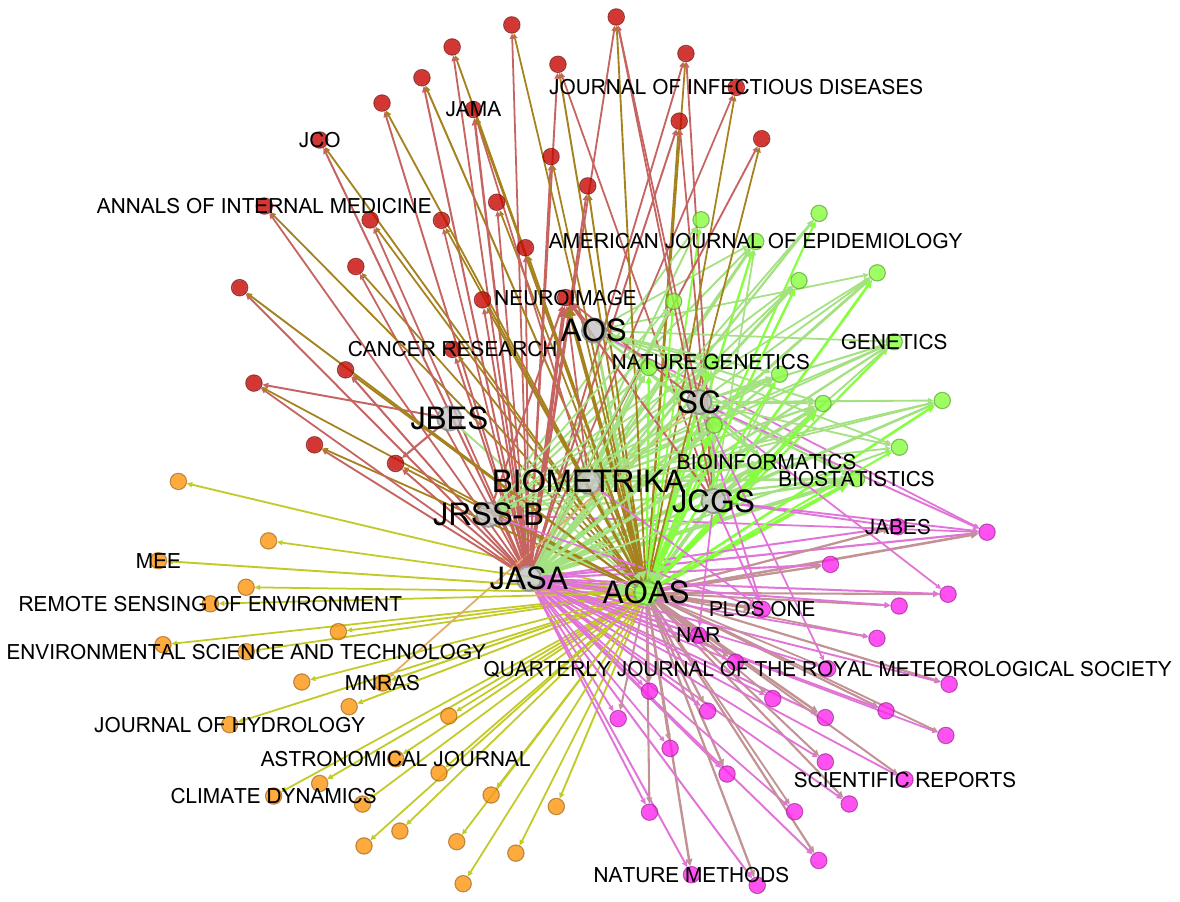}
\caption{Subnetwork of the ``Applied Statistics" community with subcommunities detected by Bi-SCORE. The 4 identified subcommunities are visualized in different colors: Interdisciplinary Research (pink), Medical Science (brown), Natural Science (orange), and Biostatistics (light green).}
\label{fig:sub_dectection}
\end{figure}

\begin{itemize}

    \item \textbf{Interdisciplinary Research (Pink)}. This subcommunity contains 28 journals, where statistical methods are applied to complex systems in various academic fields. Research in this subcommunity focuses on the practical applications conducted in an interdisciplinary context. The main journals in this subcommunity are \emph{PLoS One}, \emph{Nucleic Acids Research (NAR)}, and \emph{Journal of Agricultural, Biological and Environmental Statistics (JABES)}, which have high weighted in-degree in this subcommunity.

    \item \textbf{Medical Science (Brown)}. This subcommunity contains 28 journals, with a main focus on medical imaging, neuroscience, clinical research, and healthcare management. Research in this subcommunity emphasizes the use of statistical modeling to advance medical technologies and improve patient outcomes. Representative journals in this subcommunity include \emph{NeuroImage}, \emph{Journal of the American Medical Association (JAMA)}, and \emph{Journal of Clinical Oncology (JCO)}.

    \item \textbf{Natural Science (Orange)}. This subcommunity contains 25 journals, which primarily focus on astronomy, ecology, environmental science, and geoscience, as shown in Table \ref{table:natural science journal}. These citation patterns highlight the central role of statistical methods in natural science applications. Specifically, \emph{Monthly Notices of the Royal Astronomical Society (MNRAS)}, \emph{Methods in Ecology and Evolution (MEE)}, and \emph{Atmospheric Environment} have the highest weighted in-degree within this subcommunity.

    \item \textbf{Biostatistics (Light Green)}. This subcommunity contains 17 journals, which emphasize biology-related fields such as biochemistry, biotechnology, and genetics. Research in this community focuses on biological discovery through the application of statistical methodologies in the life sciences. The most representative journal in this community is \emph{AOAS}, as expected for a core journal. In addition, \emph{Bioinformatics}, \emph{Biostatistics}, and \emph{Nature Genetics} also have high weighted in-degree.

\end{itemize}

Figure~\ref{fig:applied_citation} reveals the knowledge sources for core journals in the 4 identified subcommunities. The following three main phenomena are observed. 
First, both \emph{AOAS} and \emph{JASA} cited heavily from the ``Interdisciplinary Research" and ``Medical Science" subcommunities. This reflects their emphasis on interdisciplinary applications of statistics, particularly in domains that combine methodological development with real-world problems. In addition, driven by the growing demand for data analysis in clinical research and health management, the application of statistical methods in medical science has grown substantially \citep{murdoch2013inevitable}. 
Second, except for \emph{AOAS}, all other core journals rarely cite from the ``Natural Science" subcommunity. This is because \emph{AOAS} aims to provide a forum for all areas of applied statistics, including nature science. In contrast, other core journals focus primarily on computational applications, methodological development, or the theory of mathematical statistics. 
Third, almost all core journals cite extensively from the ``Biostatistics" subcommunity. This reveals the prominence of biostatistics as one of the most active areas in contemporary applied statistical research \citep{efron2005bayesians}. Specifically, \emph{JBES} primarily targets topics in economics and finance, which explains its limited citation of biostatistics-related work.

\begin{figure}[htbp]
\centering
\includegraphics[width=0.9\textwidth]{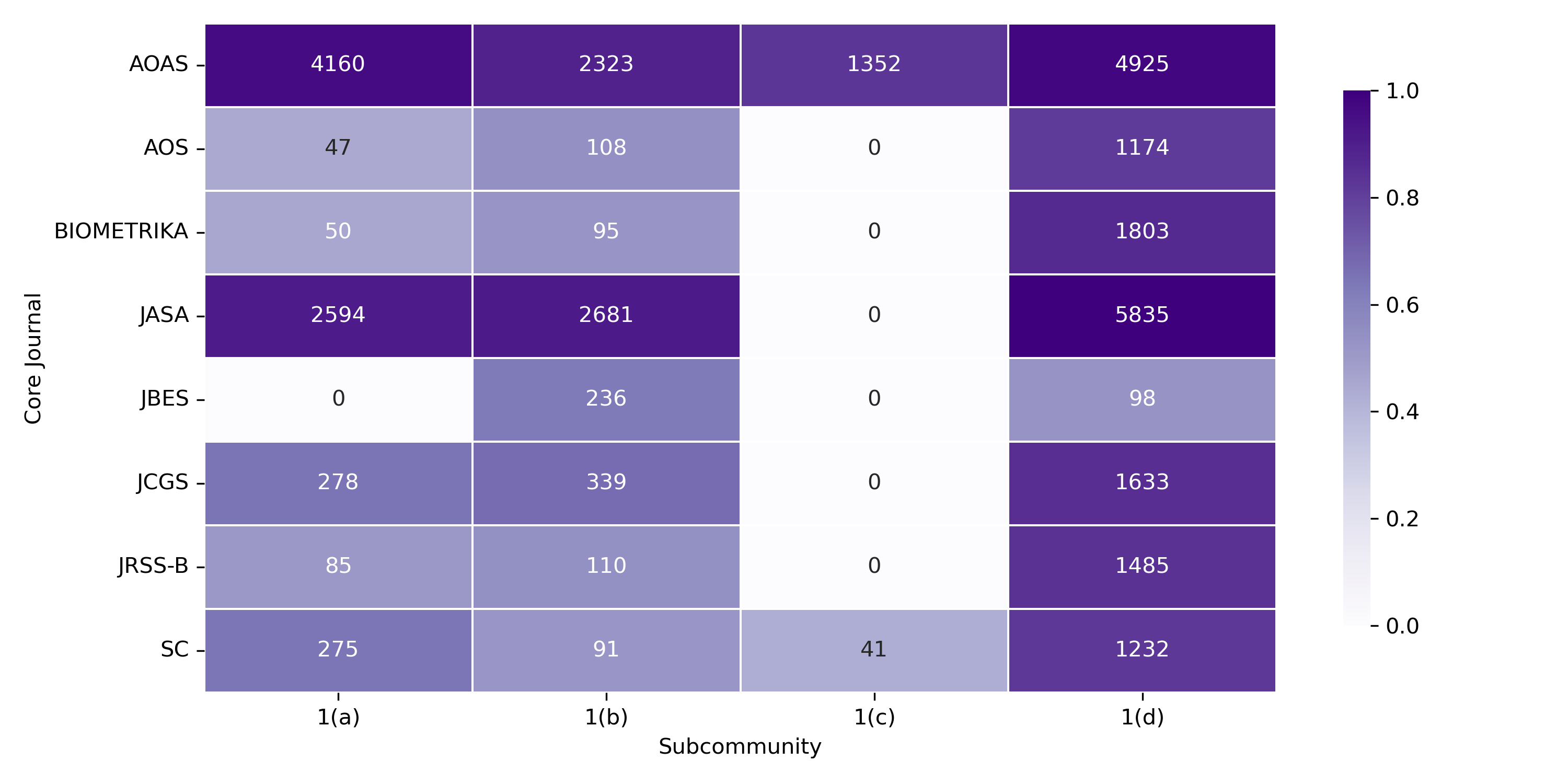}
\caption{Heatmap of knowledge sources for core journals across the 4 identified subcommunities in the ``Applied Statistics" community. The x-axis labels correspond to the 4 subcommunities: (a) Interdisciplinary Research, (b) Medical Science, (c) Natural Science, and (d) Biostatistics. Each cell indicates the number of journals from a given subcommunity cited by a specific core journal. Darker colors indicate higher citation counts (based on log-transformed and normalized values). The figure reveals distinctive citation patterns and research emphases of the core journals.}
\label{fig:applied_citation}
\end{figure}

\section{Discussion}\label{sec6}

In this study, we construct a weighted bipartite journal citation network using publication data from eight core journals over 2001–2023. To address the challenges posed by degree heterogeneity and network sparsity, we propose a novel community detection algorithm, Bi-SCORE. The method integrates spectral decomposition with a ratio-based normalization to extract latent community structures. Theoretical analysis establishes upper bounds for the numbers of misclustered row and column nodes, respectively. Extensive simulation studies under various network settings, including balanced and imbalanced sample sizes, varying degrees of heterogeneity, and different levels of sparsity, demonstrate that Bi-SCORE consistently outperforms existing methods such as nBiSC and spectral clustering in terms of clustering accuracy. In addition, we apply Bi-SCORE to the real-world citation network and identify 6 distinct communities that correspond to major academic fields. These communities include but are not limited to ``Applied Statistics", ``Statistical Methodology", ``Computational Statistics", and ``Mathematical Statistics". Specifically, we further apply Bi-SCORE within the “Applied Statistics” community and identify 4 subcommunities. The empirical results suggest that Bi-SCORE is well-suited for knowledge source discovery in large-scale citation networks with heterogeneous structure.

Several directions for future research are possible. First, the current analysis focuses on citation data from eight core statistical journals over the period 2001–2023. While these journals represent influential sources within the field of statistics, many relevant works are published in multidisciplinary outlets which also contribute to the diffusion of statistical knowledge. Incorporating a broader range of journals may yield a more comprehensive understanding of cross-disciplinary knowledge flows. Second, the Bi-SCORE method is developed for static bipartite networks. However, journal citation networks evolve over time, and it is of interest to extend Bi-SCORE to dynamic settings to better capture the temporal evolution of community structures. Third, our method assigns each journal to a single community, but in reality, some journals span multiple research areas. Future work could explore mixed-membership models that allow journals to simultaneously belong to multiple communities. 

\begin{appendices}
\renewcommand{\thelemma}{A.\arabic{lemma}}
\section{Proofs}\label{secA}

Appendix A provides supplementary material to support the main text. Sections \ref{proof-prop-UV} - \ref{proof-V-T} show the detailed proofs of the propositions, theorem, and lemmas discussed in the main paper.

\subsection{Proof of Proposition \ref{prop-UV}}
\label{proof-prop-UV}
\begin{proof}
    We define the matrices $\mathcal{J}_{\theta} \in \mR^{n \times K}$ and $\mathcal{J}_{\gamma} \in \mR^{m \times L}$ as follows: for $1 \leq i \leq n$, $1 \leq j \leq m$, $1 \leq k \leq K$ and $1 \leq l \leq L$,
    \begin{align*}
        (\mathcal{J}_{\theta})_{ik} =
        \begin{cases}
        \displaystyle \frac{\theta_i}{\|\theta^{(k)}\|} & \text{if } c^r_i = k \\
        0 & \text{if } c^r_i \neq k
        \end{cases}
        \quad \text{and} \quad
        (\mathcal{J}_{\gamma})_{jl} =
        \begin{cases}
        \displaystyle \frac{\gamma_j}{\|\boldsymbol{\gamma}^{(l)}\|} & \text{if } c^c_j = l \\
        0 & \text{if } c^c_j \neq l
        \end{cases}.
    \end{align*}
    Based on these definitions, along with the definitions of $\Psi_{\theta}$, $\Psi_{\gamma}$, the expectation matrix $\Omega$ is expressed as $\Omega = (\mathcal{J}_{\theta} \|\theta\| \Psi_{\theta}) B (\mathcal{J}_{\gamma} \|\gamma\| \Psi_{\gamma})^\top$. By introducing the substitution $S \equiv \Psi_{\theta} B \Psi_{\gamma}^\top$, we can simplify the expression for $\Omega$ to
    \begin{align}
    \label{Omega-decomp}
        \Omega = \|\theta\|\|\gamma\| \mathcal{J}_{\theta} S \mathcal{J}_{\gamma}^\top.
    \end{align}

    Since the diagonal matrices $\Psi_{\theta}$ and $\Psi_{\gamma}$ are of full rank, the rank of $S$ is given by $\text{rank}(S) = \text{rank}(\Psi_{\theta} B \Psi_{\gamma}^\top) = \kappa$. We then define the compact SVD of the $K \times L$ matrix $S$ as 
    \begin{align}
    \label{S-svd}
        S = Y \Lambda_S H^\top,
    \end{align}
    where $\Lambda_S$ is a $\kappa \times \kappa$ rectangular diagonal matrix containing the non-negative singular values of $S$ in descending order. The matrices $Y \in \mR^{K \times \kappa}$ and $H \in \mR^{L \times \kappa}$ are orthogonal matrices.

    Substituting the SVD of $S$ from \eqref{S-svd} into the expression for $\Omega$ from \eqref{Omega-decomp}, we obtain:
    \begin{align}
    \label{Omega-JYJH}
        \Omega = \|\theta\| \|\gamma\| (\mathcal{J}_{\theta} Y) \Lambda_S (H \mathcal{J}_{\gamma})^\top.
    \end{align}
    From the definitions of $\mathcal{J}_{\theta}$ and $\mathcal{J}_{\gamma}$, it follows that they possess orthonormal columns, specifically $\mathcal{J}_{\theta}^\top \mathcal{J}_{\theta} = I_{K \times K}$ and $\mathcal{J}_{\gamma}^\top \mathcal{J}_{\gamma} = I_{L \times L}$. Consequently,
    \begin{align}
    \label{JYJH}
        (\mathcal{J}_{\theta} Y)^\top (\mathcal{J}_{\theta} Y) = Y^\top \mathcal{J}_{\theta}^\top \mathcal{J}_{\theta} Y = I_{K \times K},  \nonumber \\
        (\mathcal{J}_{\gamma} H)^\top (\mathcal{J}_{\gamma} H) = H^\top \mathcal{J}_{\gamma}^\top \mathcal{J}_{\gamma} H = I_{L \times L}.
    \end{align}
    By \eqref{JYJH}, we observe that $\mathcal{J}_{\theta} Y$ and $\mathcal{J}_{\gamma} H$ have orthogonal columns. Thus, \eqref{Omega-JYJH} constitutes the compact SVD of the matrix $\Omega$. Denoting the compact SVD of $\Omega$ as $\Omega = U \Lambda_{\Omega} V^\top$, we can identify
    \begin{align}
    \label{UVLambda}
        U = \mathcal{J}_{\theta} Y,~
        V = \mathcal{J}_{\gamma} H,~\text{and}~
        \Lambda_{\Omega} = \|\theta\| \|\gamma\| \Lambda_S.
    \end{align}
    From \eqref{UVLambda}, we directly derive the relationships given in \eqref{singularvalue}, \eqref{rowform}, and \eqref{columnform}. As $Y$ is an orthogonal matrix, its rows are orthonormal, meaning $\|Y_{c^r_i \cdot}\| = 1$. Consequently, $\|U_{i \cdot}\| = \|(\theta_i / \|\theta^{(c_i^r)}\|) Y_{c^r_i \cdot}\| = \theta_i / \|\theta^{(c_i^r)}\|$. Furthermore, based on \eqref{thetaktheta}, we have $\|U_{i \cdot}\| \asymp \theta_i / \|\theta\|$. A similar derivation yields $\|V_{j \cdot}\| \asymp \gamma_j / \|\gamma\|$.
    This completes the proof of Proposition \ref{prop-UV}.
\end{proof}

\subsection{Proof of Proposition \ref{prop-U-UC}}
\label{proof-prop-U-UC}
Proposition \ref{prop-U-UC} aims to bound the distance between the singular vectors of $\Omega$ and those of $A$. To achieve this, we first establish several key lemmas. These include Lemma \ref{lambda-Omega}, which demonstrates that the eigenvalues of $\Omega$ are on the order of $\|\theta\|^2 \|\gamma\|^2$; Lemma \ref{lemma-A-Omega}, which bounds the distance between the random adjacency matrix $A$ and its expected value $\Omega$; and Lemma \ref{lemmaSsvd}, which provides a positive lower bound for the singular value gap $\lambda_1(S^\top S) - \lambda_2(S^\top S)$. By combining these foundational lemmas, we then apply the Davis-Kahan Theorem (Lemma \ref{lemma-UO-U}) to establish Proposition \ref{prop-U-UC}.

\begin{lemma}
\label{lambda-Omega}
    Under the weighted bipartite DCBM and in conjunction with Assumption \ref{assump-theta-gamma}, for $1 \leq k \leq \kappa$, the eigenvalues of $\Omega^\top \Omega$ (and $\Omega \Omega^\top$) are of the order: 
    \begin{align*}
        \lambda_k(\Omega^\top \Omega) =\lambda_k(\Omega \Omega^\top) \asymp \|\theta\|^2 \|\gamma\|^2.
    \end{align*}
\end{lemma}
The proof can be found in Appendix \ref{proof-lambda-Omega}.

\begin{lemma}
    \label{lemma-A-Omega}
    Under Assumption \ref{assump-B}, for sufficiently large $n$ and $m$, with probability at least $1-1/n - 1/m$, the spectral norm distance between the random adjacency matrix $A$ and its expectation $\Omega$ is bounded as:
    \begin{align*}
        \|A - \Omega\| \leq \sqrt{2 \log(nm) \mathcal{Z}} + \log(nm)/3.
    \end{align*}
\end{lemma}
The proof can be found in Appendix \ref{proof-lemma-A-Omega}.

\begin{lemma}
\label{lemmaSsvd}
    For a matrix $S = Y \Lambda_S H^\top$, and for $k=1,\dots,K$ and $l=1,\dots,L$, under Assumptions \ref{assump-B} and \ref{assump-theta-gamma}, the following properties hold:
    \begin{gather}
        0 < C \leq Y_{k1} \leq 1 \quad \text{and} \quad 0 < C \leq H_{l1} \leq 1 \nonumber \\
        \lambda_1(S^\top S) - \lambda_2(S^\top S) = \lambda_1(S S^\top) - \lambda_2(S S^\top)  \geq C, \nonumber  \\
        U_{i1}>0,~ V_{j1}>0 \quad \text{for} \quad 1 \leq i \leq n~\text{and}~1 \leq j \leq m. \nonumber 
    \end{gather}
\end{lemma}
The proof of Lemma \ref{lemmaSsvd} can be found in Appendix \ref{proof-lemmaSsvd}.

\begin{lemma}[Theorem 2 in \cite{yu2015useful}]
    \label{lemma-UO-U}
    Let $G, \widehat{G} \in \mathbb{R}^{n \times n}$ be symmetric matrices with eigenvalues $\lambda_1 \geq \cdots \geq \lambda_n$ and $\widehat{\lambda}_1 \geq \cdots \geq \widehat{\lambda}_n$, and corresponding eigenvectors $d_1, \dots, d_n$ and $\widehat{d}_1, \dots, \widehat{d}_n$, respectively. Fix indices $1 \leq r \leq s \leq n$. Assume that the eigenvalue gap condition holds: $\min(\lambda_{r-1} - \lambda_r, \lambda_s - \lambda_{s+1}) \geq 0$, where we define $\lambda_0 = \infty$ and $\lambda_{n+1} = -\infty$. Let $k=s-r+1$, $D = (d_r, d_{r+1}, \dots, d_s) \in \mathbb{R}^{n \times k}$, and define the matrices  $\widehat{D} = (\widehat{d}_r, \widehat{d}_{r+1}, \dots, \widehat{d}_s) \in \mathbb{R}^{n \times k}$. Then, there exists an orthogonal matrix $O \in \mathbb{R}^{k \times k}$ such that
    \begin{align*}
        \|DO - \widehat{D}\|_F \leq \frac{2^{\frac{3}{2}} k^{\frac{1}{2}} \|G -\widehat{G}\|}{ \min(\lambda_{r-1} - \lambda_r, \lambda_s - \lambda_{s+1}) }.
    \end{align*}
\end{lemma}

\begin{proof}[Proof of Proposition \ref{prop-U-UC}]
    We begin by bounding the spectral norm of the difference between $A^\top A$ and $\Omega^\top \Omega$. First, with a probability at least $1-1/n - 1/m$, we have
    \begin{align}
        \|A^\top A - \Omega^\top \Omega\| &\leq \|A^\top A - A^\top \Omega\| + \|A^\top \Omega - \Omega^\top \Omega\| \leq \|A\| \|A - \Omega\| + \|A-\Omega\|\|\Omega\| \nonumber \\
        &\leq \|A-\Omega\|(\|A\| + \|\Omega\|) \leq \|A-\Omega\|(\|A-\Omega\|+2\|\Omega\|) \nonumber \\
        &\leq \left\{\sqrt{2 \log(nm) \mathcal{Z}} + \log(nm)/3\right\} \left\{ 2\|\Omega\| + \sqrt{2 \log(nm) \mathcal{Z}} + \log(nm)/3  \right\} \nonumber \\
        &\leq C_1 \left\{\sqrt{\log(nm) \mathcal{Z}} + \log(nm)\right\} \|\theta\| \|\gamma\| + C_2 \left\{\sqrt{\log(nm) \mathcal{Z}} + \log(nm)\right\}^2,
        \label{eq-AA-OmegaOmega}
    \end{align}
    where the fifth step follows from Lemma \ref{lemma-A-Omega}, and the sixth step follows from Lemma \ref{lambda-Omega}, which implies $\|\Omega\| = \sqrt{\lambda_1 (\Omega^\top \Omega)} \asymp \|\theta\|\|\gamma\|$.

    Next, we apply Lemma \ref{lemma-UO-U} to bound the distance between the top left singular vectors. We have
    \begin{align*}
        \|\widehat{U}_{\cdot 1} - U_{\cdot 1} C_U\| &\leq \frac{C\|AA^\top - \Omega \Omega^\top\|}{\lambda_1(\Omega \Omega^\top) - \lambda_2(\Omega \Omega^\top)} \\
        &\leq \frac{C_1 \left\{\sqrt{\log(nm) \mathcal{Z}} + \log(nm)\right\} \|\theta\| \|\gamma\| + C_2 \left\{\sqrt{\log(nm) \mathcal{Z}} + \log(nm)\right\}^2}{\lambda_1(\Omega \Omega^\top) - \lambda_2(\Omega \Omega^\top)} \\
        & \leq \frac{C_1 \left\{\sqrt{\log(nm) \mathcal{Z}} + \log(nm)\right\} \|\theta\| \|\gamma\| + C_2 \left\{\sqrt{\log(nm) \mathcal{Z}} + \log(nm)\right\}^2}{C \|\theta\|^2 \|\gamma\|^2} \\
        & \leq C_1 \frac{\sqrt{\log(nm) \mathcal{Z}} + \log(nm)}{\|\theta\| \|\gamma\|} + C_2 \left\{ \frac{\sqrt{\log(nm) \mathcal{Z}} + \log(nm)}{\|\theta\| \|\gamma\|} \right\}^2 \\
        & \leq C \frac{\sqrt{\log(nm) \mathcal{Z}} + \log(nm)}{\|\theta\| \|\gamma\|},
    \end{align*}
    where the second step follows from \eqref{eq-AA-OmegaOmega}, the third step follows from \eqref{singularvalue} in Proposition \ref{prop-UV}, and Lemma \ref{lemmaSsvd}, which implies that
    \begin{align}
        \lambda_1(\Omega \Omega^\top) - \lambda_2(\Omega \Omega^\top) = \|\theta\|^2 \|\gamma\|^2 (\lambda_1(SS^\top) - \lambda_2(SS^\top)) \geq C \|\theta\|^2 \|\gamma\|^2,
        \label{eq-lambda-Omega-Omega}
    \end{align}
    and the last step follows from \eqref{eq-lognmZ}, and thus on the right-hand side of the inequality, the first term dominates the second term for large $n$ and $m$.

    Similarly, we apply Lemma \ref{lemma-UO-U} to bound the Frobenius norm distance between the subspaces spanned by the 2nd to $\kappa$-th left singular vectors.
    \begin{align*}
        \|\widehat{U}_{\cdot 2 \sim \kappa} - U_{\cdot 2 \sim \kappa}O_U \|_F &\leq \frac{C\|AA^\top - \Omega \Omega^\top\|}{\min( \lambda_1(\Omega \Omega^\top) - \lambda_2(\Omega \Omega^\top), 
        \lambda_{\kappa}(\Omega \Omega^\top) - \lambda_{\kappa +1}(\Omega \Omega^\top))} \\
        &\leq \frac{C\|AA^\top - \Omega \Omega^\top\|}{\min( \lambda_1(\Omega \Omega^\top) - \lambda_2(\Omega \Omega^\top), 
        \lambda_{\kappa}(\Omega \Omega^\top))} \\
        &\leq \frac{C_1 \left\{\sqrt{\log(nm) \mathcal{Z}} + \log(nm)\right\} \|\theta\| \|\gamma\| + C_2 \left\{\sqrt{\log(nm) \mathcal{Z}} + \log(nm)\right\}^2}{C \|\theta\|^2 \|\gamma\|^2} \\
        & \leq C_1 \frac{\sqrt{\log(nm) \mathcal{Z}} + \log(nm)}{\|\theta\| \|\gamma\|} + C_2 \left\{ \frac{\sqrt{\log(nm) \mathcal{Z}} + \log(nm)}{\|\theta\| \|\gamma\|} \right\}^2 \\
        & \leq C \frac{\sqrt{\log(nm) \mathcal{Z}} + \log(nm)}{\|\theta\| \|\gamma\|},
    \end{align*}
    where the second step follows from Proposition \ref{prop-UV}, which implies that $\lambda_{(\kappa + 1)}(\Omega \Omega^\top) = 0$, the third step follows from Lemma \ref{lambda-Omega}, \eqref{eq-AA-OmegaOmega} and \eqref{eq-lambda-Omega-Omega}.

    Following an analogous procedure for the right singular vectors, we can obtain similar bounds for $\widehat{V}$ and $V$: 
    \begin{align*}
        \|\widehat{V}_{\cdot 1} - V_{\cdot 1} C_V\| \leq C \frac{\sqrt{\log(nm) \mathcal{Z}} + \log(nm)}{\|\theta\| \|\gamma\|}, \\
        \|\widehat{V}_{\cdot 2 \sim \kappa} - V_{\cdot \sim \kappa} C_V\|_F \leq C \frac{\sqrt{\log(nm) \mathcal{Z}} + \log(nm)}{\|\theta\| \|\gamma\|}.
    \end{align*}
    This completes the proof.
\end{proof}

\subsection{Proof of Proposition \ref{proposition R-R}}
\label{proof-proposition R-R}
\begin{proof}
    For two row nodes $i_1$ and $i_2$, the distance is  
    \begin{align}
    \label{Rr-Rr}
        \|R^r_{i_1 \cdot} - R^r_{i_2 \cdot} \|^2 &= \Big\|\frac{(U_{\cdot 2 \sim \kappa} O_{U})_{i_1 \cdot}}{C_U U_{i_1 1}} - \frac{(U_{\cdot 2 \sim \kappa} O_{U})_{i_2 \cdot}}{C_U U_{i_2 1}}\Big\|^2 = \Big\| \frac{(U_{\cdot 2 \sim \kappa})_{i_1 \cdot}}{U_{i_1 1}} - \frac{(U_{\cdot 2 \sim \kappa})_{i_2 \cdot}}{U_{i_2 1}} \Big\|^2 \nonumber \\
        &= \Big\| \frac{(U_{\cdot 2 \sim \kappa})_{i_1 \cdot}}{U_{i_1 1}} - \frac{(U_{\cdot 2 \sim \kappa})_{i_2 \cdot}}{U_{i_2 1}} \Big\|^2 + \Big\| \frac{U_{i_1 1}}{U_{i_1 1}} - \frac{U_{i_2 1}}{U_{i_2 1}} \Big\|^2 \nonumber \\
        &= \Big\| \frac{U_{i_1 \cdot}}{U_{i_1 1}} - \frac{U_{i_2 \cdot}}{U_{i_2 1}} \Big\|^2  = \left\| \frac{\frac{\theta_{i_1}}{\|\theta^{(c_{i_1}^r)}\|} Y_{c^r_{i_1} \cdot}}{\frac{\theta_{i_1}}{\|\theta^{(c_{i_1}^r)}\|} Y_{c^r_{i_1} 1}} - \frac{\frac{\theta_{i_2}}{\|\theta^{(c_{i_2}^r)}\|} Y_{c^r_{i_2} \cdot}}{\frac{\theta_{i_2}}{\|\theta^{(c_{i_2}^r)}\|} Y_{c^{i_2} 1}} \right\|^2 \nonumber \\
        &= \left\|\frac{Y_{c^r_{i_1} \cdot}}{Y_{c^r_{i_1} 1}} - \frac{Y_{c^r_{i_2} \cdot}}{Y_{c^{i_2} 1}} \right\|^2,
    \end{align}
    where the second step is derived from Proposition \ref{prop-U-UC}, and the fifth equality follows from Proposition \ref{prop-UV}. Consequently, \eqref{Rr-Rr} shows that $\|R^r_{i_1 \cdot} - R^r_{i_2 \cdot}\|^2 = 0$ if $c^r_{i_1} = c^r_{i_2}$. Otherwise, if $c^r_{i_1} \neq c^r_{i_2}$, we obtain
    \begin{align}
        \|R^r_{i_1 \cdot} - R^r_{i_2 \cdot} \|^2 &= \left\|\frac{Y_{c^r_{i_1} \cdot}}{Y_{c^r_{i_1} 1}} - \frac{Y_{c^r_{i_2} \cdot}}{Y_{c^{i_2} 1}} \right\|^2 = \left\|\frac{Y_{c^r_{i_1} \cdot}}{Y_{c^r_{i_1} 1}} \right\|^2 + \left\|\frac{Y_{c^r_{i_2} \cdot}}{Y_{c^{i_2} 1}} \right\|^2 - 2 \left  \langle \frac{Y_{c^r_{i_1} \cdot}}{Y_{c^r_{i_1} 1}}, \frac{Y_{c^r_{i_2} \cdot}}{Y_{c^{i_2} 1}} \right \rangle \nonumber \\
        &= \frac{1}{\lvert Y_{c^r_{i_1} 1} \rvert} + \frac{1}{\lvert Y_{c^r_{i_2} 1} \rvert}  \geq 1+1 = 2,
    \end{align}
    where the third step holds because $H$ is an orthogonal matrix. The fourth step follows from Lemma \ref{lemmaSsvd}.

    Similarly, it can be demonstrated that $\|R^c_{j_1 \cdot} - R^c_{j_2 \cdot}\|^2 \geq 2$ when $c^c_{j_1} \neq c^c_{j_2}$, and is zero otherwise. This concludes the proof.
\end{proof}

\subsection{Proof of Proposition \ref{R-R-F}}
\label{proof-R-R-F}
To prove Proposition \ref{R-R-F}, we first establish Lemma \ref{V-T} to bound the number of ill-behaviour nodes, and then present a useful technical inequality in Lemma \ref{vaub}.

For a constant $0 < C <1$, we define the following sets
\begin{align}
    \widehat{T}^{r} \equiv \left( 1 \leq i \leq n, \left| \frac{\widehat{U}_{i1}}{C_U U_{i1}} -1 \right| \leq C \right) \quad \text{and} \quad \widehat{T}^{c} \equiv \left( 1 \leq j \leq m, \left| \frac{\widehat{V}_{j1}}{C_V V_{j1}} -1 \right| \leq C \right)
    \label{hat-T}
\end{align}
After this, Lemma \ref{V-T} helps us figure out how many nodes are not in these sets $\widehat{T}^{r}$ and $\widehat{T}^{c}$.

\begin{lemma}
    \label{V-T}
    For nodes in $\widehat{T}^{r}$ and $\widehat{T}^{c}$, and under Assumptions \ref{assump-B} and \ref{assump-theta-gamma}, the following equations hold
    \begin{align*}
        \lvert \widehat{U}_{i1} \rvert \asymp \lvert C_U U_{i1} \rvert \asymp \frac{\theta_i}{\|\theta\|}~ \text{for} ~ i \in \widehat{T}^{r}, \\
        \lvert \widehat{V}_{j1} \rvert \asymp \lvert C_V V_{j1} \rvert \asymp \frac{\gamma_j}{\|\gamma\|}~ \text{for} ~ j \in \widehat{T}^{c}.
        \end{align*}
        Furthermore, with a probability at least $1-1/n-1/m$, the cardinality of $\mathcal{V}^r \setminus \widehat{T}^{r}$ and $\mathcal{V}^c \setminus \widehat{T}^{c}$ satisfy
        \begin{align*}
            | \mathcal{V}^r \setminus \widehat{T}^{r} | \leq \frac{C\left\{\sqrt{\log(nm) \mathcal{Z}} + \log(nm)\right\}^2}{\theta_{\min}^2 \|\gamma\|^2} ~ \text{and} ~  |\mathcal{V}^c \setminus \widehat{T}^c| \leq \frac{C\left\{\sqrt{\log(nm) \mathcal{Z}} + \log(nm)\right\}^2}{\gamma_{\min}^2 \|\theta\|^2}.
        \end{align*}
\end{lemma}
The proof of Lemma \ref{V-T} can be found in Appendix \ref{proof-V-T}.

\begin{lemma}[Lemma A.6 in \cite{wang2020spectral}]
\label{vaub}
    For any two vectors $u,v \in \mathbb{R}^n$, and any two positive numbers $a,b \in \mathbb{R}$, the following inequality holds,
    \begin{align*}
        \left\| \frac{u}{a} - \frac{v}{b} \right\|^2 \leq 2\left\{ \frac{1}{a^2} \|u-v\|^2 + \frac{(b-a)^2}{(ab)^2} \|v\|^2 \right\}.
    \end{align*}
\end{lemma}

\begin{proof}[Proof of Proposition \ref{R-R-F}]
    We observe that
    \begin{align}
        \|(U_{\cdot 2 \sim \kappa} O_U)_{i \cdot}\|^2 = \|(U_{\cdot 2 \sim \kappa})_{i \cdot} O_U\|^2 = \|(U_{\cdot 2 \sim \kappa})_{i \cdot}\|^2 \leq \|U_{i \cdot}\|^2 \leq C \frac{\theta_i^2}{\|\theta\|^2},
        \label{UOU}
    \end{align}
    where the last inequality comes from Proposition \ref{prop-UV}. Next, we want to show $\|R^r_{i \cdot}\|^2 \leq C$. Based on the definition of $R^r$, we have
    \begin{align}
    \label{R-2norm}
        \|R^r_{i \cdot}\|^2 = \left\|  \frac{(U_{\cdot 2 \sim \kappa} O_{U})_{i \cdot}}{C_U U_{i1}} \right\|^2 \leq \frac{C \frac{\theta_i^2}{\|\theta\|^2}}{\lvert C_U U_{i1} \rvert^2} \leq \frac{C \frac{\theta_i^2}{\|\theta\|^2}}{\frac{\theta_i^2}{\|\theta^{(c^r_i)}\|^2} \lvert Y_{c^r_i 1} \rvert^2} \leq C,
    \end{align}
    where the third step follows from Proposition \ref{prop-UV}, and the fourth step uses Lemma \ref{lemmaSsvd}, which implies $Y_{k1} \geq C >0$. 

    Next, we want to bound the total squared Frobenius norm difference between our estimated representation $\widehat{R}^r$ and the true representation $R^r$. We split this total difference into two parts: the sum over ``ill-behaved'' nodes (those outside $\widehat{T}^r$) and the sum over ``well-behaved'' nodes (those inside $\widehat{T}^r$).
    \begin{align}
        \|\widehat{R}^r - R^r\|_F^2 = \sum_{ i \in (\mathcal{V}^r \setminus \widehat{T}^r)} \|\widehat{R}^r_{i \cdot}- R^r_{i \cdot}\|^2 + \sum_{ i \in \widehat{T}^r} \|\widehat{R}^r_{i \cdot}- R^r_{i \cdot}\|^2.
        \label{hatR-R}
    \end{align}
    For the first part, we have
    \begin{align}
        \sum_{ i \in (\mathcal{V}^r \setminus \widehat{T}^r)} \|\widehat{R}^r_{i \cdot}- R^r_{i \cdot}\|^2 &\leq C \sum_{ i \in (\mathcal{V}^r \setminus \widehat{T}^r)} (\|\widehat{R}^r_{i \cdot}\|^2 + \|R^r_{i \cdot}\|^2) \leq C \sum_{ i \in (\mathcal{V}^r \setminus \widehat{T}^r)} (\kappa \tau_n^2 + C) \nonumber \\ 
        & \leq C |\mathcal{V}^r \setminus \widehat{T}^r| \tau_n^2  \leq \frac{C \tau_n^2 \left\{\sqrt{\log(nm) \mathcal{Z}} + \log(nm)\right\}^2}{\theta_{\min}^2 \|\gamma\|^2},
        \label{firstterm}
    \end{align}
    where the second step uses equations \eqref{hat-R} and \eqref{R-2norm} to bound the individual terms. The third step states that as $n$ grows large enough, the term $\tau_n$ becomes dominant compared to the constant $C$. The last step applies Lemma \ref{V-T}.

    For the second part in \eqref{hatR-R}, we have
    \begin{align*}
        \sum_{ i \in \widehat{T}^r} \|\widehat{R}^r_{i \cdot}- R^r_{i \cdot}\|^2 &\leq C \sum_{ i \in \widehat{T}^r} \left\| \frac{\widehat{U}_{i 2 \sim \kappa}}{\widehat{U}_{i1}} - \frac{U_{i 2 \sim \kappa} O_U}{C_U U_{i1}} \right\|^2 \\
        &\leq C \sum_{ i \in \widehat{T}^r} \left\{ \frac{1}{\widehat{U}_{i1}^2} \|\widehat{U}_{i 2 \sim \kappa} - U_{i 2 \sim \kappa} O_U\|^2 + \frac{(C_U U_{i1} - \widehat{U}_{i1})^2}{(\widehat{U}_{i1} C_U U_{i1})^2} \|U_{i 2 \sim \kappa} O_U\|^2 \right\} \\
        &\leq C \sum_{ i \in \widehat{T}^r} \left\{ \frac{\|\theta\|^2}{\theta_i^2} \|\widehat{U}_{i 2 \sim \kappa} - U_{i 2 \sim \kappa} O_U\|^2 + \frac{\|\theta\|^2}{\theta_i^2} (C_U U_{i1} - \widehat{U}_{i1})^2 \right\} \\
        & \leq C \frac{\|\theta\|^2}{\theta_{\min}} \sum_{ i \in \widehat{T}^r} \left\{ \|\widehat{U}_{i 2 \sim \kappa} - U_{i 2 \sim \kappa} O_U\|^2 + (C_U U_{i1} - \widehat{U}_{i1})^2 \right\} \\
        &\leq C \frac{\|\theta\|^2}{\theta_{\min}} \left( \|\widehat{U}_{\cdot 2 \sim \kappa} - U_{\cdot 2 \sim \kappa} O_U\|_F^2 + \|\widehat{U}_{\cdot 1} - C_U U_{\cdot 1}\|^2 \right) \\
        &\leq \frac{C \left\{ \sqrt{\log(nm) \mathcal{Z}} + \log(nm) \right\}^2 }{\|\gamma\|^2 \theta_{\min}^2},
        \label{secondterm}
    \end{align*}
    where the first step follows from the fact that $\|R_{i \cdot}^r\| \leq C$ (see \eqref{R-2norm}), and $\tau_n$ scales with $n$, which implies $\tau_n \geq C \geq \|R_{i \cdot}^r\|$ for $n$ large enough. Thus, although \eqref{hat-R} shows that the elements in $R_{i \cdot}^r$ are truncated by $\tau_n$, we still have $ \|\widehat{R}^r_{i \cdot}- R^r_{i \cdot}\|^2 \leq \left\| \frac{\widehat{U}_{i 2 \sim \kappa}}{\widehat{U}_{i1}} - \frac{U_{i 2 \sim \kappa} O_U}{C_U U_{i1}} \right\|^2$ for large $n$, the second step uses Lemma \ref{vaub}, the third step uses Lemma \ref{V-T} and \eqref{UOU}, and the last step follows from Proposition \ref{prop-U-UC}.

    By combining the bounds from \eqref{firstterm} and \eqref{secondterm}, we get the total bound for the row representations $\|\widehat{R}^r - R^r\|_F^2 \leq \frac{C \tau_n^2 \left\{\sqrt{\log(nm) \mathcal{Z}} + \log(nm)\right\}^2}{\theta_{\min}^2 \|\gamma\|^2}$. A similar line of reasoning applies to the column representations, leading to $\|\widehat{R}^c - R^c\|_F^2 \leq \frac{C \tau_m^2 \left\{\sqrt{\log(nm) \mathcal{Z}} + \log(nm)\right\}^2}{\gamma_{\min}^2 \|\theta\|^2}$. This completes the proof.
\end{proof}

\subsection{Proof of Proposition \ref{X-R}}
\label{proof-X-R}
\begin{proof}
    Recall that $\widehat{X}^r$ is defined as 
    \begin{align*}
        \widehat{X}^r = \underset{X^r \in \mathcal{X}_{n,K-1,K}}{\argmin} \left\| X^r - \widehat{R}^r \right\|_F^2, \quad \widehat{X}^c = \underset{X^c \in \mathcal{X}_{m,K-1,L}}{\argmin} \left\| X^c - \widehat{R}^c \right\|_F^2.
    \end{align*}
    Note that $R^r$ and $R^c$ also belong to $\mathcal{X}_{n,K-1,K}$ and $\mathcal{X}_{m,K-1,L}$, respectively. Therefore, 
    \begin{align*}
        \|\widehat{X}^r - \widehat{R}^r\| \leq \|R^r - \widehat{R}^r\|, \quad \|\widehat{X}^c - \widehat{R}^c\| \leq \|R^c - \widehat{R}^c\|.
    \end{align*}
    From this, we have 
    \begin{align*}
        \|\widehat{X}^r - R^r\|_F^2 &\leq \|\widehat{X}^r -\widehat{R}^r + \widehat{R}^r - R^r\|_F^2 \leq C\|\widehat{X}^r -\widehat{R}^r\|_F^2 + C\|\widehat{R}^r - R^r\|_F^2 \\
        & \leq C\|R^r -\widehat{R}^r\|_F^2 + C\|\widehat{R}^r - R^r\|_F^2 \leq C\|\widehat{R}^r - R^r\|_F^2 \\
        &\leq \frac{C \tau_n^2 \left\{\sqrt{\log(nm) \mathcal{Z}} + \log(nm)\right\}^2}{\theta_{\min}^2 \|\gamma\|^2},
    \end{align*}
    where the last step follows from Proposition \ref{R-R-F}.
    Similarly, we can show that 
    $$
    \|\widehat{X}^c - R^c\|_F^2 \leq \frac{C \tau_m^2 \left\{\sqrt{\log(nm) \mathcal{Z}} + \log(nm)\right\}^2}{\gamma_{\min}^2 \|\theta\|^2}.
    $$
    This completes the proof of Proposition \ref{X-R}.
\end{proof}

\subsection{Proof of Theorem \ref{theorem1}}
\label{Proof_of_Theorem_1}
\begin{proof}
    For row nodes $i_1$ and $i_2$ in set $\mathcal{W}^r$, if they belong to different communities, then 
\begin{align*}
    \|\widehat{X}_{i_1 \cdot}^r - \widehat{X}_{i_2 \cdot}^r\| = & \|\widehat{X}_{i_1 \cdot}^r - R_{i_1 \cdot}^r + R_{i_1 \cdot}^r - R_{i_2 \cdot}^r + R_{i_2 \cdot}^r - \widehat{X}_{i_2 \cdot}^r\| \\
    \geq & \| R_{i_1 \cdot}^r - R_{i_2 \cdot}^r \| - \| \widehat{X}_{i_1 \cdot}^r - R_{i_1 \cdot}^r + R_{i_2 \cdot}^r - \widehat{X}_{i_2 \cdot}^r \| \\
    \geq & \| R_{i_1 \cdot}^r - R_{i_2 \cdot}^r \| - \|\widehat{X}_{i_1 \cdot}^r - R_{i_1 \cdot}^r\| - \|\widehat{X}_{i_2 \cdot}^r - R_{i_2 \cdot}^r\|.
\end{align*}
Similarly, for column nodes $j_1$ and $j_2$ in set $\mathcal{W}^c$, if they are in different communities, we have 
\begin{align*}
    \|\widehat{X}_{j_1 \cdot}^c - \widehat{X}_{j_2 \cdot}^c\| \geq \| R_{j_1 \cdot}^c - R_{j_2 \cdot}^c \| - \|\widehat{X}_{j_1 \cdot}^c - R_{j_1 \cdot}^c\| - \|\widehat{X}_{j_2 \cdot}^c - R_{j_2 \cdot}^c\|.
\end{align*}
By Proposition \ref{proposition R-R} and the definition of $\mathcal{W}^r$ and $\mathcal{W}^c$, for $i_1,~i_2 \in \mathcal{W}^r$ and $j_1,~j_2 \in \mathcal{W}^c$, we have
\begin{align*}
    \|\widehat{X}_{i_1 \cdot}^r - \widehat{X}_{i_2 \cdot}^r\| \geq 1, \quad \|\widehat{X}_{j_1 \cdot}^c - \widehat{X}_{j_2 \cdot}^c\| \geq 1.
\end{align*}
Therefore, if nodes $i_1$ and $i_2$ belong to different communities, their corresponding rows in $\widehat{X}^r$ are sufficiently different. Since we assume that $\lvert \mathcal{V}^r \setminus \mathcal{W}^r \rvert < \min\{n_1, \dots, n_K\}$, the set $\mathcal{W}^r$ contains at least one node from each community. Given that $\widehat{X}^r$ has exactly $K$ distinct rows, we conclude that if two nodes in $\mathcal{W}^r$ are in the same community, their rows in $\widehat{X}^r$ are identical; otherwise, their rows are sufficiently different. Thus, nodes in $\mathcal{W}^r$ are clustered correctly. By the same reasoning, nodes in $\mathcal{W}^c$ are also clustered correctly. From the definitions of $\mathcal{W}^r$, $\mathcal{W}^c$ and Proposition \ref{X-R}, we have
\begin{align*}
    \lvert \mathcal{V}^r \setminus \mathcal{W}^r \rvert & < 4 \sum_{i \in (\mathcal{V}^r \setminus \mathcal{W}^r)} \|\widehat{X}^r_{i \cdot} - R^r_{i \cdot}\|^2  \leq 4 \|\widehat{X}^r - R^r\|_F^2 \\
    & \leq \frac{C \tau_n^2 \left\{\sqrt{\log(nm) \mathcal{Z}} + \log(nm)\right\}^2}{\theta_{\min}^2 \|\gamma\|^2}.
\end{align*}
Similarly, we obtain
\begin{align*}
    \lvert \mathcal{V}^c \setminus \mathcal{W}^c \rvert < \frac{C \tau_m^2 \left\{\sqrt{\log(nm) \mathcal{Z}} + \log(nm)\right\}^2}{\gamma_{\min}^2 \|\theta\|^2}.
\end{align*}
This completes the proof.
\end{proof}

\subsection{Proof of Lemma \ref{lambda-Omega}}
\label{proof-lambda-Omega}

\begin{proof}[Proof of Lemma \ref{lambda-Omega}]
    According to Proposition \ref{prop-UV}, we have $\sigma_k(\Omega) = \|\theta\| \|\gamma\| \sigma_k(S)$, for $1 \leq k \leq \kappa$. Therefore, it is sufficient to show $0 < C_1 \leq \sigma_{\kappa}(S) \leq \sigma_1(S) \leq C_2$.

    Recall $S = \Psi_{\theta} B \Psi_{\gamma}^\top$, where $\Psi_{\theta}$ and $\Psi_{\gamma}$ are diagonal matrices. The terms $\|\Psi_{\theta}\|$ and $\|\Psi_{\theta}\|_{\min}$ represent the largest and smallest absolute values of the diagonal entries of $\Psi_{\theta}$, respectively. From \eqref{thetaktheta}, we have
    \begin{gather*}
        \|\Psi_{\theta}\| = \max_k (\Psi_{\theta})_{kk} = \max_k \frac{\|\theta^{(k)}\|}{\|\theta\|} \leq C, \\
        \|\Psi_{\theta}\|_{\min} = \min_k (\Psi_{\theta})_{kk} = \min_k \frac{\|\theta^{(k)}\|}{\|\theta\|} \geq C.
    \end{gather*}
    Thus, there exist two constants $C_m > 0$ and $C_M > 0$ such that 
    \begin{align}
        \label{Psi-theta}
        C_m \leq \|\Psi_{\theta}\|_{\min} \leq \|\Psi_{\theta}\| \leq C_M.
    \end{align}
    A similar argument shows that
    \begin{align}
        \label{Psi-gamma}
        C_m \leq \|\Psi_{\gamma}\|_{\min} \leq \|\Psi_{\gamma}\| \leq C_M.
    \end{align}
    From the definition of $S$, we have 
    \begin{align}
    \label{S-max}
        \sigma_1(S) = \|S\| = \|\Psi_{\theta} B \Psi_{\gamma}^\top\| \leq \|\Psi_{\theta}\| \|B\| \|\Psi_{\gamma}^\top\| \leq C,
    \end{align}
    because $B$ is a constant matrix (i.e., it does not change with $n$), and therefore there exists a constant $C$, such that $\|B\| \leq C$. On the other hand, 
    \begin{align}
    \label{S-min}
        \sigma_{\kappa}(S) = \|S\|_{\min} = \|\Psi_{\theta} B \Psi_{\gamma}^\top\|_{\min} \geq \|\Psi_{\theta}\|_{\min} \|B\|_{\min} \|\Psi_{\gamma}^\top\|_{\min} \geq C,
    \end{align}
    where the last inequality follows from Assumption \ref{assump-B}.

    Combining \eqref{S-max} and \eqref{S-min}, we get $\sigma_{k} \asymp C$ for $k=1,\dots,\kappa$. Therefore, we have
    $$
    \lambda_k(\Omega^\top \Omega)=\lambda_k(\Omega \Omega^\top) = \sigma_k^2(\Omega) = \sigma_k^2(S) \|\theta\|^2 \|\gamma\|^2 \asymp \|\theta\|^2 \|\gamma\|^2.
    $$
    This completes the proof of Lemma \ref{lambda-Omega}.
\end{proof}

\subsection{Proof of Lemma \ref{lemma-A-Omega}}
\label{proof-lemma-A-Omega}
We first introduce the following lemma.
\begin{lemma}[Corollary 4 in \cite{bacry2018concentration}]
\label{lemma-poisson-centrality}
    Let $A$ be a $n \times m$ random matrix whose entries $A_{ij}$ follow a Poisson distribution with mean $\Omega_{ij}$. Then, for any $x > 0$, we have
    \begin{align*}
        P\left(\|A - \Omega\|_2 \geq \sqrt{2 \max(\|\Omega\|_{1,\infty}, \|\Omega\|_{\infty, 1}) x} + \frac{x}{3} \right) \leq (n+m)e^{-x}
    \end{align*}
    where $\Omega = (\Omega_{ij})$, $\|\Omega\|_{1,\infty} = \max_i\|\Omega_{i \cdot}\|_1$, and $\|\Omega\|_{\infty, 1} = \max_j\|\Omega_{\cdot j}\|_1$.
\end{lemma}

\begin{proof}[Proof of Lemma \ref{lemma-A-Omega}]
    From Assumption \ref{assump-B}, we have $\Omega_{ij} = \theta_i \gamma_j B_{c_i^r c_j^c} \leq \theta_i \gamma_j$. Therefore,
    \begin{align*}
        \|\Omega\|_{1,\infty} = \max_i\|\Omega_{i \cdot}\|_1 \leq \max_i \theta_i \|\gamma\|_1, \\
        \|\Omega\|_{\infty,1} = \max_j\|\Omega_{\cdot j}\|_1 \leq \max_j \gamma_j \|\theta\|_1.
    \end{align*}
    Let $\mathcal{Z} = \max(\theta_{\max}, \gamma_{\max}) \max(\|\theta\|_1, \|\gamma\|_1)$, where $\theta_{\max} = \max_i \theta_i$ and $\gamma_{\max} = \max_j \gamma_j$. Then,
    $\max(\|\Omega\|_{1,\infty}, \|\Omega\|_{\infty, 1}) \leq \mathcal{Z}$.
    Applying Lemma \ref{lemma-poisson-centrality} with $x = \log(nm)$, we obtain that, with probability at least $1-1/n - 1/m$, 
    \begin{align*}
        \|A - \Omega\| \leq \sqrt{2 \log(nm) \mathcal{Z}} + \log(nm)/3.
    \end{align*}
    This completes the proof.
\end{proof}

\subsection{Proof of Lemma \ref{lemmaSsvd}}
\label{proof-lemmaSsvd}
We first present the following lemma.
    \begin{lemma}[Theorem 8.4.4 in \cite{horn2012matrix}, Lemma A.7 in \cite{wang2020spectral}] 
    \label{lemma-nonnegative}
        For every $K \times K$ irreducible nonnegative, and positive semidefinite matrix $M$, let $U_{\cdot 1}$ be the eigenvector corresponding to its largest eigenvalue. Then, 
        \begin{itemize}
            \item[(i)] $U_{\cdot 1}$ can be a positive vector.
            \item[(ii)] The largest eigenvalue is an algebraically simple eigenvalue.
        \end{itemize}
    \end{lemma}
\begin{proof}[Proof of Lemma \ref{lemmaSsvd}]
    To prove Lemma \ref{lemmaSsvd}, it is sufficient to show that $S^\top S$ and $SS^\top$ are irreducible and nonnegative, and these properties do not depend on $n$ or $m$. Once this is established, part (i) of Lemma \ref{lemma-nonnegative} implies that $Y_{k1} \geq C >0$ and $H_{l1} \geq C > 0$, where $Y_{\cdot 1}$ and $H_{\cdot 1}$ are the eigenvectors corresponding to $SS^\top$ and $S^\top S$, respectively. Moreover, $Y_{k1} \geq C >0$  leads to $U_{i1} > 0$, since $U_{i \cdot} = (\theta_i / \|\theta^{(c_i^r)}\|) Y_{c_i^r \cdot}$ for $i=1,\dots,n$. Similarly, we have $V_{j1} >0$ for $j=1,\dots,m$. Part (ii) of Lemma \ref{lemma-nonnegative} further gives $\lambda_1(S^\top S) - \lambda_2(S^\top S) = \lambda_1(S S^\top) - \lambda_2(S S^\top) \geq C$. This completes the proof of Lemma \ref{lemmaSsvd}.

    Next, we verify that $S^\top S$ and $SS^\top$ are irreducible and nonnegative. Recall that $S=\Psi_{\theta} B \Psi_{\gamma}^\top$. From \eqref{Psi-theta}, \eqref{Psi-gamma}, and the definition of the diagonal matrices $\Psi_\theta$ and $\Psi_\gamma$, we see that there exist constants $C_1$ and $C_2$ such that $0 < C_1 \leq \Psi_{\theta,kk} \leq C_2$ and $0 < C_1 \leq \Psi_{\gamma,ll} \leq C_2$. Therefore, for $k=1,\dots,K$ and $l=1,\dots,L$,
    \begin{align*}
        C_1^2 B_{kl} \leq S_{kl} \leq C_2^2 B_{kl}.
    \end{align*}
    For $1 \leq l_1, l_2 \leq L$, we have
    \begin{align*}
        (S^\top S)_{l_1 l_2} = \sum_{k = 1}^K (S^\top)_{l_1 k} S_{k, l_2} \leq C_2^4 \sum_{k = 1}^K (B^\top)_{l_1 k} B_{k, l_2} = C (B^\top B)_{l_1 l_2}.
    \end{align*}
    Similarly, for $1 \leq l_1, l_2 \leq L$, we obtain
    \begin{align*}
        C(B^\top B)_{l_1 l_2} \leq (S^\top S)_{l_1 l_2}.
    \end{align*}
    Since $B$ is a constant matrix with positive entries, $B^\top B$ is irreducible and nonnegative by Assumption \ref{assump-B}. Hence, $S^\top S$ is also nonnegative and irreducible, and these properties do not depend on $n$ or $m$. By the same reasoning, $SS^\top$ is nonnegative and irreducible, and these properties do not depend on $n$ or $m$. This completes the proof.
    
\end{proof}

\subsection{Proof of Lemma \ref{V-T}}
\label{proof-V-T}
\begin{proof}[Proof of Lemma \ref{V-T}]
    We first prove $|C_U U_{i1}| \asymp |\theta_i| / \|\theta\|$ for $i=1,\dots,n$. By Proposition \ref{prop-UV}, we have $U_{i \cdot} = \frac{\theta_i}{\|\theta^{(c^r_i)}\|} Y_{c^r_i \cdot}$, and thus $C_U U_{i1} = C_U \frac{\theta_i}{\|\theta^{(c^r_i)}\|} Y_{c^r_i 1}$, where $|C_U| = 1$ by Proposition \ref{prop-U-UC}. From \eqref{thetaktheta} and Lemma \ref{lemmaSsvd}, we obtain
    \begin{align}
        |C_U U_{i1}| \asymp \left| C_U \frac{\theta_i}{\|\theta^{(c^r_i)}\|} Y_{c^r_i 1} \right| \asymp \frac{|\theta_i|}{\|\theta\|}, \quad \text{for}~i=1,\dots,n.
        \label{CU-U}
    \end{align}
    Next, by \eqref{hat-T}, nodes in the set $\widehat{T}^r$ satisfies $\left| \frac{\widehat{U}_{i1}}{C_U U_{i1}} -1 \right| \leq C < 1$, and therefore $|\widehat{U}_{i1}| \asymp |C_U U_{i1}|$. Thus, for $i \in \widehat{T}^r$  we have
    \begin{align*}
        |\widehat{U}_{i1}|\asymp |C_U U_{i1}|\asymp \frac{|\theta_i|}{\|\theta\|}. 
    \end{align*}
    Similarly, $|\widehat{V}_{j1}|\asymp |C_V V_{j1}|\asymp \frac{|\gamma_j|}{\|\gamma\|}$ for $j \in \widehat{T}^c$.

    Furthermore, by \eqref{CU-U} and \eqref{theta-min} in Assumption \ref{assump-theta-gamma}, we have $|C_U U_{i1}| \asymp \frac{|\theta_i|}{\|\theta\|} >0$, so using $C_U U_{i1}$ as the denominator for all $i$ is valid. We then derive
    \begin{align*}
        & \sum_{ i \in (\mathcal{V}^r \setminus \widehat{T}^{r})} \left( \frac{\widehat{U}_{i1}}{C_U U_{i1}} -1 \right)^2 = \sum_{ i \in (\mathcal{V}^r \setminus \widehat{T}^{r})} \left( \frac{1}{C_U U_{i1}} \right)^2 (\widehat{U}_{i1} - C_U U_{i1})^2 \\
        &\leq \sum_{ i \in (\mathcal{V}^r \setminus \widehat{T}^{r})} \frac{\|\theta\|^2}{\theta_{\min}^2}(\widehat{U}_{i1} - C_U U_{i1})^2 \leq \sum_{i=1}^n \frac{\|\theta\|^2}{\theta_{\min}^2}(\widehat{U}_{i1} - C_U U_{i1})^2 \\
        &\leq \frac{\|\theta\|^2}{\theta_{\min}^2} \|\widehat{U}_{\cdot 1} - U_{\cdot 1} C_U\|^2 \leq C\frac{\left\{\sqrt{\log(nm) \mathcal{Z}} + \log(nm)\right\}^2}{\theta_{\min}^2 \|\gamma\|^2},
    \end{align*}
    where the second step follows from \eqref{CU-U}, and the last step follows from Proposition \ref{prop-U-UC}. 

    Since nodes in the set $\mathcal{V}^r \setminus \widehat{T}^r$ satisfy $\{\widehat{U}_{i1}/(C_U U_{i1}) -1\}^2 \geq C_0^2$, we have
    \begin{align*}
        |\mathcal{V}^r \setminus \widehat{T}^r| = \sum_{i \in (\mathcal{V} \setminus \widehat{T}^r)} 1 \leq \sum_{i \in (\mathcal{V} \setminus \widehat{T}^r)} \frac{1}{C_0^2} \left( \frac{\widehat{U}_{i1}}{C_U U_{i1}} - 1 \right)^2 \leq C\frac{\left\{\sqrt{\log(nm) \mathcal{Z}} + \log(nm)\right\}^2}{\theta_{\min}^2 \|\gamma\|^2}.
    \end{align*}
    Similarly, we can show that 
    \begin{align*}
        |\mathcal{V}^c \setminus \widehat{T}^c| \leq C\frac{\left\{\sqrt{\log(nm) \mathcal{Z}} + \log(nm)\right\}^2}{\gamma_{\min}^2 \|\theta\|^2}.
    \end{align*}
    This completes the proof.
\end{proof}

\end{appendices}

\begin{appendices}
\setcounter{section}{1}
\section{Additional Results}\label{secB}

Appendix B presents additional results from our analysis of the journal citation network. 

\label{add-table}
\begin{table}[!htbp]
\centering
\rotatebox{90}{ 
\begin{minipage}{\textheight} 
\caption{Eight representative journals. Journal name, number of publications and JCR category are presented.}
\label{table:journal}
\begin{tabularx}{\textheight}{>{\raggedright\arraybackslash}p{7.5cm} c X}
\toprule
\multicolumn{1}{c}{\textbf{Journal Name}} & \multicolumn{1}{c}{\textbf{Number of Publications}} & \multicolumn{1}{c}{\textbf{JCR Category}} \\
\midrule
Annals of Applied Statistics (AOAS) & 1677 & STATISTICS \& PROBABILITY \\
Annals of Statistics (AOS)          & 2464 & STATISTICS \& PROBABILITY \\
Biometrika                           & 1804 & 
\begin{tabular}[t]{@{}l@{}}
BIOLOGY\\
MATHEMATICAL \& COMPUTATIONAL BIOLOGY\\
STATISTICS \& PROBABILITY
\end{tabular}\\
Journal of the Royal Statistical Society Series B--statistical Methodology (JRSS-B) & 1129 & STATISTICS \& PROBABILITY \\
Journal of the American Statistical Association (JASA) & 3185 & STATISTICS \& PROBABILITY \\
Journal of Computational and Graphical Statistics (JCGS) & 1614 & STATISTICS \& PROBABILITY \\
Journal of Business \& Economic Statistics (JBES) & 1271 & 
\begin{tabular}[t]{@{}l@{}}
ECONOMICS\\
SOCIAL SCIENCES, MATHEMATICAL METHODS\\
STATISTICS \& PROBABILITY
\end{tabular}\\
Statistics and Computing (SC)       & 1613 & 
\begin{tabular}[t]{@{}l@{}}
COMPUTER SCIENCE, THEORY \& METHODS\\
STATISTICS \& PROBABILITY
\end{tabular}\\
\bottomrule
\end{tabularx}
\end{minipage}
}
\end{table}

\begin{table}[htbp]
\centering
\caption{Top 10 references for {\it core journals}. Journal name abbreviation is presented with its weight in parentheses.}\label{table:top 10 references}
\begin{tabularx}{\textwidth}{@{}p{2.2cm} p{3.5cm} p{2.2cm} p{4cm}@{}}
\toprule
\textbf{Journal Name} & \textbf{Top 10 References} & \textbf{Journal Name} & \textbf{Top 10 References} \\
\midrule
\textbf{AOAS} & JASA (4189) & \textbf{JASA} & JASA (12294) \\
 & JRSS-B (2125) & & AOS (8653) \\
 & AOS (1751) & & Biometrika (5396) \\
 & Biometrika (1686) & & JRSS-B (5122) \\
 & Biometrics (1638) & & Biometrics (3086) \\
 & AOAS (1191) & & Econometrica (1600) \\
 & STAT MED (857) & & STAT MED (1514) \\
 & Bioinformatics (797) & & JOE (1392) \\
 & PNAS (794) & & STAT SINICA (1382) \\
 & Nature (712) & & JMLR (1183) \\
\midrule
\textbf{AOS} & AOS (14471) & \textbf{JCGS} & JASA (4366) \\
 & JASA (4922) & & AOS (2795) \\
 & JRSS-B (2993) & & JRSS-B (2549) \\
 & Biometrika (2859) & & Biometrika (1994) \\
 & JMA (1380) & & JCGS (1829) \\
 & Bernoulli (1279) & & Biometrics (839) \\
 & TIT (1253) & & CSDA (755) \\
 & AMS (1180) & & JMLR (740) \\
 & JMLR (1152) & & SC (704) \\
 & STAT SINICA (1060) & & STAT SINICA (650) \\
\midrule
\textbf{Biometrika} & Biometrika (4998) & \textbf{JBES} & JOE (4411) \\
 & JASA (4368) & & Econometrica (3929) \\
 & AOS (4176) & & JASA (2340) \\
 & JRSS-B (2518) & & JBES (1894) \\
 & Biometrics (1442) & & AOS (1650) \\
 & STAT MED (688) & & ECONOMET THEOR (1008) \\
 & STAT SINICA (643) & & JOF (987) \\
 & Econometrica (571) & & AER (971) \\
 & STAT SCI (506) & & JAE (916) \\
 & JMA (480) & & RES (876) \\
\midrule
\textbf{JRSS-B} & AOS (3618) & \textbf{SC} & JASA (3273) \\
 & JASA (3593) & & JRSS-B (2501) \\
 & JRSS-B (2542) & & AOS (2317) \\
 & Biometrika (2320) & & Biometrika (1869) \\
 & Biometrics (831) & & SC (1535) \\
 & Econometrica (522) & & CSDA (1163) \\
 & STAT SCI (493) & & JCGS (1036) \\
 & JMLR (491) & & Biometrics (736) \\
 & STAT SINICA (472) & & JMLR (727) \\
 & ARXIV (420) & & ARXIV (658) \\
\botrule
\end{tabularx}
\end{table}

\begin{table}[htbp]
\centering
\caption{Journals in the ``Natural Science'' subcommunity. Journal name and weight are presented, where the weight refers to the total citation counts from the core journals.}\label{table:natural science journal}
\begin{tabularx}{\textwidth}{p{8cm} >{\centering\arraybackslash}X}
\toprule
\multicolumn{1}{c}{\textbf{Journal Name}} & \multicolumn{1}{c}{\textbf{Weight}} \\
\midrule
Monthly Notices of the Royal Astronomical Society (MNRAS) & 145 \\
Methods in Ecology and Evolution (MEE) & 83 \\
Astronomy and Astrophysics (A\&A) & 80 \\
Atmospheric Environment & 79 \\
PLoS Medicine (PLoS MED) & 62 \\
Journal of Hydrology & 61 \\
Analytical Chemistry & 60 \\
Demography & 59 \\
Astronomical Journal (AJ) & 58 \\
Journal of Neurophysiology & 57 \\
Philosophical Transactions of the Royal Society B & 49 \\
BMC Genomics & 45 \\
Remote Sensing of Environment (RSE) & 45 \\
Biophysical Journal & 44 \\
Climate Dynamics & 44 \\
Journal of Evolutionary Economics & 44 \\
PLoS Biology (PLoS Biol) & 44 \\
Journal of Theoretical Biology & 43 \\
Statistical Methodology & 43 \\
Blood & 42 \\
Clinical Cancer Research & 42 \\
Journal of Neuroscience Methods & 42 \\
Environmental Science and Technology & 41 \\
Nature Reviews Neuroscience & 41 \\
Cancer Cell & 40 \\
\bottomrule
\end{tabularx}
\end{table}

\end{appendices}

\clearpage
\bibliography{sn-bibliography}

\end{document}